\documentclass[11pt,peerreview,draftcls,onecolumn]{IEEEtran}

\newif\ifsubmission
\submissionfalse

\usepackage[latin1]{inputenc}
\usepackage{amsmath}
\usepackage{amsfonts}
\usepackage{amssymb}
\usepackage{amstext}
\usepackage{amsthm}
\usepackage{graphicx}
\usepackage{url}
\usepackage{color}
\usepackage{bm}
\usepackage{dsfont}
\usepackage{color}
\usepackage{xcolor}

\newtheorem{theorem}{Theorem}
\newtheorem{lemma}{Lemma}

\newtheorem{corollary}{Corollary}
\newtheorem{definition}{Definition}

\newcommand{\fig}[1]{Figure~\ref{#1}}

\newcommand{\BB}{\mathcal B}

\newcommand{\DD}{\mathcal D}
\newcommand{\XX}{\mathcal X}
\newcommand{\MM}{\mathcal M}

\renewcommand{\SS}{\mathcal S}
\renewcommand{\AA}{\mathcal A}

\newcommand{\LL}{\mathcal L}
\newcommand{\PP}{\mathcal P}

\newcommand{\A}{\mathcal A}
\newcommand{\sMM}{
    {\scalebox{0.8}{$\scriptscriptstyle\mathcal{M}$}}%\scriptscriptstyle
  }

\newcommand{\rarrow}{\tiny \to}

\newcommand{\SIa}{$SI_{c \textnormal{-} tr}^{a}$}
\newcommand{\SIat}{$SI_{c \textnormal{-} tr}^{a,t}$}
\newcommand{\SIr}{$SI_{c \textnormal{-} tr}^{r}$}

\colorlet{mygreen}{green!60!gray}

\newcommand{\BTcomm}[1]{\textcolor{mygreen}{{#1}}}
\newcommand{\MB}[1]{\textcolor{orange}{{#1}}}
\newcommand{\IC}[1]{\textcolor{cyan}{{#1}}}

\begin{document}

%\title{Binary Hypothesis Testing with Corrupted Training Data}
%\title{Distinguishability of Sources of Information under Corrupted Training Data}

%\title{Distinguishability of  Information Sources under Corrupted Training Data}

%\title{Adversarial Learning
%and Testing: Distinguishability of Information Sources under Corrupted Training Data}

\title{Adversarial Source Identification Game with Corrupted Training}

\author{Mauro Barni, \IEEEmembership{Fellow, IEEE}, Benedetta Tondi, \IEEEmembership{Student Member, IEEE}
\thanks{M. Barni is with the Department of Information Engineering and Mathematics, University of Siena, Via Roma 56, 53100 - Siena, ITALY, phone: +39 0577 234850 (int. 1005), e-mail: barni@dii.unisi.it; B. Tondi is with the Department of Information Engineering and Mathematics, University of Siena, Via Roma 56, 53100 - Siena, ITALY, e-mail: benedettatondi@gmail.com.}
}

%\markboth{DRAFT}%
%{M. Barni, B. Tondi: Binary Hypothesis Testing with Corrupted Training Data}
%\markboth{IEEE TRANSACTIONS ON INFORMATION THEORY}%
\markboth{IEEE TRANSACTIONS ON INFORMATION THEORY, ~Vol.~X, No.~X, XXXXXXX~XXXX}
{M. Barni, B. Tondi: Adversarial Source Identification Game with Corrupted Training}

\maketitle

\begin{abstract}
We study a variant of the source identification game with training data in which part of the training data is corrupted by an attacker. In the addressed scenario, the defender aims at deciding whether a test sequence has been drawn according to a discrete memoryless source $X \sim P_X$, whose statistics are known to him through the observation of a training sequence generated by $X$. In order to  undermine the correct decision under the alternative hypothesis that the test sequence has not been drawn from $X$, the attacker can modify a sequence produced by a source $Y \sim P_Y$ up to a certain distortion, and corrupt the training sequence either by adding some fake samples or by replacing some samples with fake ones. We derive the unique rationalizable equilibrium of the two versions of the game in the asymptotic regime and by assuming that the defender bases its decision by relying only on the first order statistics of the test and the training sequences. By mimicking Stein's lemma, we derive the best achievable performance for the defender when the first type error probability is required to tend to zero exponentially fast with an arbitrarily small, yet positive, error exponent. We then use such a result to analyze the ultimate distinguishability of any two sources as a function of the allowed distortion and the fraction of corrupted samples injected into the training sequence.
\end{abstract}


\IEEEpeerreviewmaketitle

\medmuskip=0mu
\thinmuskip=0mu
\thickmuskip=0mu

\section{Introduction}
\label{sec.intro}

%\BTcomm{oppure 'Source Identification and Distinguishability under Corrupted Training Data'}
\IEEEPARstart{A}{dversarial} Signal Processing (AdvSP) is an emerging discipline aiming at modelling the interplay between a defender wishing to carry out a certain processing task, and an attacker aiming at impeding it \cite{BarGon13}. Binary decision in an adversarial setup is one of the most recurrent problems in AdvSP, due to its importance in many application scenarios. Among binary decision problems, source identification is one of the most studied subjects, since it lies at the heart of several security-oriented disciplines, like multimedia forensics, anomaly detection, traffic monitoring, steganalysis and so on.

The source identification game has been introduced in \cite{BT13} to model the interplay between the defender and the attacker by resorting to concepts drawn from game and information theory. According to the model put forward in \cite{BT13}, the defender and the attacker have a perfect knowledge of the to-be-distinguished sources. In \cite{BTtit} the analysis is pushed a step forward by considering a scenario in which the sources are known only through the observation of a training sequence. Finally, \cite{BT_SMargin} introduces the security margin concept, a synthetic parameter characterising the ultimate distinguishability of two sources under adversarial conditions.

In this paper, we extend the analysis further, by considering a situation in which the attacker may interfere with the learning phase by corrupting part of the training sequence. Adversarial learning is a rather novel concept, which has been studied for some years from a machine learning perspective \cite{machineLEARN,Barreno2010, Roli15}. Due to the natural vulnerability of machine learning systems, in fact, the attacker may take an important advantage if no countermeasures are adopted by the defender. The use of a training sequence to gather information about the statistics of the to-be-distinguished sources can be seen as a very simple learning mechanism, and the analysis of the impact that an attack carried out in such a phase has on the performance of a decision system may help shedding new light on this important problem.
To be specific, we extend the game-theoretic framework introduced in \cite{BTtit} and \cite{BT_SMargin} to model a situation in which the attacker is given the possibility of corrupting part of the training sequence. By adopting a game-theoretic perspective, we derive the optimal strategy for the defender and the optimal corruption strategy for the attacker when the length of the training sequence and the observed sequence tends to infinity. Given such optimum strategies, expressed in the form of game equilibrium point, we analyse the best achievable performance when the type I and II error probabilities tend to zero exponentially fast. Specifically, we study the distinguishability of the sources as a function of the fraction of training samples corrupted by the attacker and when the test sequence can be modified up to a certain distortion level. The results of the analysis are summarised in terms of blinding corruption level, defined as the fraction of corrupted samples making a reliable distinction between the two sources impossible, and security margin, defined as the maximum distortion of the observed sequence for which a reliable distinction is possible (see \cite{BT_SMargin}). The analysis is applied to two different scenarios wherein the attacker is allowed respectively to {\em add} a certain amount of fake samples to the training sequence and to selectively {\em replace} a fraction of the samples of the training sequences with fake samples. As we will see, the second case is more favourable to the attacker, since a lower distortion and a lower number of corrupted training samples are enough to prevent a correct decision.

Given the above general framework, the main results proven in this paper can be summarised as follows:

\begin{enumerate}
\item{We rigorously define the source identification game with addition of corrupted training samples (\SIa~game) and show that such a game is a dominance solvable game admitting an asymptotic equilibrium point when the length of the training and test sequences tend to infinity (Theorem \ref{teo.aseq} and following discussion in Section \ref{sec.SI_CTR_add});}
\item{We evaluate the payoff of the game at the equilibrium and derive the expression of the indistinguishability region, defined as the region with the sources $Y$ which can not be distinguished from $X$ because of the attack (Theorems \ref{theo_as_payoff_si} and \ref{theo_indist_reg_SIa}, Section \ref{sec.SI_CTR_add});}
\item{Given any two sources $X$ and $Y$, we derive the security margin and the blinding corruption level defined as the maximum distortion introduced into the test sequence and maximum fraction of fake training samples introduced by the attacker, still allowing the distinction of $X$ and $Y$ while ensuring positive error exponents for the two kinds of errors of the test (Theorem \ref{theorem_EMD_L} and Definition \ref{def.SM} in Section \ref{sec.SM_SIa});}
\item{We repeat the entire analysis for the source identification game with selective replacement of training samples (\SIr~game), and compare the two versions of the game (Theorem \ref{teo.SIR} and subsequent discussion in Section \ref{sec.SI_CTR_c})}.
\item{The main proofs of the paper rely on a generalised version of Sanov's theorem \cite{CandT,Dembo2009}, which is proven in Appendix \ref{sec.appendix.Sanov}. In fact, Theorem \ref{theo.extended_Sanov}, and its use to simplify some of the proofs in the paper, can be seen as a further methodological contribution of our work.}
\end{enumerate}

This paper considerably extends the analysis presented in \cite{BTwifs14}, by providing a formal proof of the results anticipated in \cite{BTwifs14}\footnote{We also give a more precise formulation of the problem, by correcting some inaccuracies present in \cite{BTwifs14}.} and make a step forward by studying a more complex corruption scenario in which the attacker has the freedom to replace a given percentage of the training samples rather than simply adding some fake samples to the original training sequence.

The paper is organised as follows. Section \ref{sec.symbols} summarises the notation used throughout the paper, gives some definitions and introduces some basics concept of Game theory that will be used in the sequel. Section \ref{sec.SI_CTR_add} gives a rigorous definition of the \SIa~game, explaining the rationale behind the various assumptions made in the definition. In Section \ref{sec.SI_CTR_add_solution}, we prove the main theorems of the paper regarding the asymptotic equilibrium point of the \SIa~game and the payoff at the equilibrium. Section \ref{sec.SM_SIa} leverages on the results proven in Section \ref{sec.SI_CTR_add_solution} to introduce the concepts of blind corruption level and security margin, and evaluating them in the setting provided by the \SIa~game. Section \ref{sec.SI_CTR_c}, introduces and solves the \SIr~game, by paying attention to compare the results of the analysis with the corresponding results of the \SIa~game. The paper ends in Section \ref{sec.conc}, with a summary of the main results proven in the paper and the description of possible directions for future work. In order to avoid burdening the main body of the paper, the most technical  details of the proofs are gathered in the Appendix.

\section{Notation and definitions}
\label{sec.symbols}

In this section, we introduce the notation and definitions used throughout the paper. We will use capital letters to indicate discrete memoryless sources (e.g. $X$). Sequences of length $n$ drawn from a source will be indicated with the corresponding lowercase letters (e.g. $x^n$); accordingly, $x_i$ will denote the $i \textnormal{-}$th element of a sequence $x^n$. The alphabet of an information source will be indicated by the corresponding calligraphic capital letter (e.g. $\XX$). The probability mass function (pmf) of a discrete memoryless source $X$ will be denoted by $P_X$. The calligraphic letter $\PP$ will be used to indicate the class of all the probability mass functions, namely, the probability simplex in $\mathds{R}^{|\XX|}$. The notation $P_X$ will be also used to indicate the probability measure ruling the emission of sequences from a source $X$, so we will use the expressions $P_X(a)$ and $P_X(x^n)$ to indicate, respectively, the probability of symbol $a \in \XX$ and the probability that the source $X$ emits the sequence $x^n$, the exact meaning of $P_X$ being always clearly recoverable from the context wherein it is used.
%Finally, we will use the notation $P_X(A)$ to indicate the probability of the event $A$ (be it a subset of $\XX$ or $\XX^n$) under the probability measure $P_X$.
We will use the notation $P_X(A)$ to indicate the probability of $A$ (be it a subset of $\XX$ or $\XX^n$) under the probability measure $P_X$. Finally, the probability of a generic will be denoted by $Pr\{\}$.

Our analysis relies extensively on the concepts of type and type class defined as follows (see \cite{CandT} and \cite{CandK} for more details). Let $x^n$ be a sequence with elements belonging to a finite alphabet $\XX$. The type $P_{x^n}$ of $x^n$ is the empirical pmf induced by the sequence $x^n$, i.e. $\forall a \in \XX, P_{x^n} (a) = \frac{1}{n} \sum_{i=1}^n \delta(x_i, a)$, where $\delta(x_i,a) = 1$ if $x_i =a$ and zero otherwise. In the following, we indicate with $\PP^n$ the set of types with denominator $n$, i.e. the set of types induced by sequences of length $n$. Given $P \in \PP^n$, we indicate with $T(P)$ the type class of $P$, i.e. the set of all the sequences in $\XX^n$ having type $P$.
We denote by $\DD(P||Q)$ the Kullback-Leibler (KL) divergence between two distributions $P$ and $Q$, defined on the same finite alphabet $\XX$ \cite{CandT}:
\begin{equation}
\label{eq.KL}
    \DD(P||Q) = \sum_{a \in \XX} P(a) \log_2 \frac{P(a)}{Q(a)}.
\end{equation}
Most of our results are expressed in terms of the generalised log-likelihood ratio function $h$ (see \cite{BTtit,Gut89,Kendall}), which for any two given sequences $x^n$ and $t^m$ is defined as:
\begin{equation}
    h(P_{x^n}, P_{t^m}) = \DD(P_{x^n} || P_{r^{n + m}}) + \frac{m}{n} \DD(P_{t^m} || P_{r^{n+ m}}),
\label{eq.h}
\end{equation}
where $P_{r^{n + m}}$ denotes the type of the sequence $r^{n+m}$, obtained by concatenating $x^n$ and $t^m$, i.e. $r^{n+m} = x^n \| t^m$. The intuitive meaning behind the above definition is that $P_{r^{n+m}}$ is the pmf which maximises the probability that a memoryless source generates two independent sequences belonging to $T(P_{x^n})$ and $T(P_{t^m})$, and that such a probability is equal to $2^{-n h(P_{x^n}, P_{t^m})}$ at the first order in the exponent (see \cite{Kendall} or Lemma 1 in \cite{BTtit}).

Throughout the paper, we will need to compute {\em limits} and {\em distances} in $\PP$. We can do so by choosing one of the many available distances defined over $\mathbb{R}^{|\XX|}$ and for which $\PP$ is a bounded set, for instance the $L_p$ distance for which we have:
\begin{equation}
    d_{L_p}(P,Q) = \bigg( \sum_{a \in \XX} |P(a) - Q(a)|^p \bigg)^{1/p}.
\label{eq.Lp_distance}
\end{equation}
Without loss of generality, we will prove all our results by adopting the $L_1$ distance, the generalisation to different $L_p$ metrics being straightforward. In the sequel, distances between pmf's in $\PP$ will be simply indicated as $d(\cdot, \cdot)$ as a shorthand for $d_{L_1}(\cdot, \cdot)$\footnote{Throughout the paper, we will use the symbol $d(\cdot, \cdot)$ to indicate both the distortion between two sequences in $\XX^n$ and the $L_1$ distance between two pmf's in $\PP$, the exact meaning being always clear from the context,}.

We also need to introduce the {\em Hausdorff distance} as a way to measure distances between subsets of a metric space \cite{munkres2000topology}.
Let $S$ be a generic space and $d$ a distance measure defined over $S$.
%
%Let $(S, d)$ be a metric space.
For any point $x \in S$ and any non-empty subset $A \subseteq S$, the distance of $x$ from the subset $A$ is defined as:
\begin{equation}
d(x,A) ~=~ \inf_{a \in A} d(a,x).
\end{equation}
Given the above definition, the Hausdorff distance between any two subsets of $S$ is defined as follows.
\begin{definition}
For any two subsets $A$ and $B$ of $S$, let us define $\delta_B(A) = \sup_{b \in B} d(b,A)$. The Hausdorff distance $\delta_H(A,B)$ between $A$ and $B$ is given by:
\begin{equation}
\delta_H(A,B) ~=~ \max\{\delta_A(B), \delta_B(A)\}.
\end{equation}
\end{definition}
\noindent If the sets $A$ and $B$ are bounded with respect to $d$, then the Hausdorff distance always takes a finite value. The Hausdorff distance does not define a true metric, but only a pseudometric, since $\delta_H(A,B) = 0$ implies that the closures of the sets $A$ and $B$ coincide, namely $\text{\em cl}(A) = \text{\em cl}(B)$, but not necessarily that $A = B$. For this reason, in order for $\delta_H$ to be a metric, we need to restrict its definition to closed subsets\footnote{Note that in this case the $\inf$ and $\sup$ operations involved in the definition of the Hausdorff distance can be replaced with $\min$ and $\max$, respectively.}.
Let then $\mathcal{L}(S)$ denote the space of  non-empty closed and limited subsets of $S$ and let  $\delta_H: \LL(S) \times \LL(S) \rightarrow [0, \infty)$. Then, the space $\LL(S)$ endowed with the Hausdorff distance is a metric space \cite{Henri99} and we can give the following definition:
\begin{definition}
Let $\{K_n\}$ be a sequence of closed and limited subsets of $S$, i.e., $K_n \in \LL(S)$ $\forall n$. We use the notation $K_n \overset{H}{\rightarrow}~ K$ to indicate that the sequence has limit in $(\LL(S), \delta_H)$ and the limiting set is $K$.
\label{def_limit_set}
\end{definition}

\subsection{Basic notions of Game Theory}

In this section, we introduce some basic notions and definitions of Game Theory.

A 2-player game is defined as a quadruple $(\SS_1,\SS_2,u_1, u_2)$, where $\SS_1 = \{s_{1,1} \dots s_{1,n_1}\}$ and $\SS_2 = \{s_{2,1} \dots s_{2,n_2}\}$ are the set of strategies the first and the second player can choose from, and $u_l(s_{1,i}, s_{2,j}), l= 1,2$, is the payoff of the game for player $l$, when the first player chooses the strategy $s_{1,i}$ and the second chooses $s_{2,j}$. A pair of strategies $(s_{1,i}, s_{2,j})$ is called a profile. When $u_1(s_{s1,i}, s_{2,j}) = -u_2(s_{1,i}, s_{2,j})$, the win of a player is equal to the loss of the other and the game is said to be a zero-sum game. The sets $\SS_1$, $\SS_2$ and the payoff functions are assumed to be known to both players. Throughout the paper we consider strategic games, i.e., games in which the players choose their strategies beforehand without knowing the  strategy chosen by the opponent player.

The final goal of game theory is to determine the existence of equilibrium points, i.e. profiles that in {\em some sense} represent the {\em best} choice for both players  \cite{Osb94}. The most famous notion of equilibrium is due to Nash.
A profile is said to be a Nash equilibrium if no player can improve its payoff by changing its strategy unilaterally. Despite its popularity, the practical meaning of Nash equilibrium is often unclear, since there is no guarantee that the players will end up playing at the equilibrium.
A particular kind of games for which stronger forms of equilibrium exist are the so called {\em dominance solvable} games \cite{Osb94}. To be specific, a strategy is said to be strictly dominant for one player if it is the best strategy for the player, i.e., the strategy which corresponds to the largest payoff, no matter how the other player decides to play. When one such strategy exists for one of the players, he will surely adopt it.
In a similar way, we say that a strategy $s_{l,i}$ is strictly dominated by strategy $s_{l,j}$, if the payoff achieved by player $l$ choosing $s_{l,i}$ is always lower than that obtained by playing $s_{l,j}$ regardless of the choice made by the other player. The recursive elimination of dominated strategies is a common technique for solving games. In the first step, all the dominated strategies are removed from the set of available strategies, since no rational player would ever play them. In this way, a new, smaller game is obtained. At this point, some strategies, that were not dominated before, may be dominated in the remaining  game, and hence are eliminated. The process goes on until no dominated strategy exists for any player. A {\em rationalizable equilibrium} is any profile which survives the iterated elimination of dominated strategies \cite{ChenGames,Bern84}. If at the end of the process only one profile is left, the remaining profile is said to be the {\em only rationalizable equilibrium} of the game. The corresponding strategies are the only rational choice for the two players and the game is said {\em dominance solvable}.

\enlargethispage{\baselineskip}

\section{Source identification game with addition of corrupted training samples (\SIa)}
\label{sec.SI_CTR_add}

In this section, we  give a rigorous definition of the Source Identification game with addition of corrupted training samples.

Given a discrete and memoryless source $X \sim P_X$ and a test sequence $v^n$, the goal of the defender (D) is to decide whether $v^n$ has been drawn from $X$ (hypothesis $H_0$) or not (alternative hypothesis $H_1$). By adopting a Neyman-Pearson perspective, we assume that D must ensure that the false positive error probability ($P_{fp}$), i.e., the probability of rejecting $H_0$ when $H_0$ holds (type I error) is lower than a given threshold. Similarly to the previous versions of the game studied in \cite{BT13} and \cite{BTtit}, we assume that D relies only on first order statistics to make a decision. For mathematical tractability, likewise earlier papers, we study the asymptotic version of the game when $n \rarrow \infty$, by requiring that $P_{fp}$ decays exponentially fast when $n$ increases, with an error exponent at least equal to $\lambda$, i.e. $P_{fp} \le 2^{-n \lambda}$. On its side, the attacker aims at increasing the false negative error probability ($P_{fn}$), i.e., the probability of accepting $H_0$ when $H_1$ holds (type II error). Specifically, A takes a sequence $y^n$ drawn from a source $Y \sim P_Y$ and modifies it in such a way that D decides that the modified sequence $z^n$ has been generated by $X$. In doing so, A must respect a distortion constraint requiring that the average per-letter distortion between $y^n$ and $z^n$ is lower than $L$.

Players A and D know the statistics of $X$ through a training sequence, however the training sequence can be partly corrupted by A. Depending on how the training sequence is modified by the attacker, we can define different versions of the game. In this paper, we focus on two possible cases: in the first case, hereafter referred to as source identification game with addition of corrupted samples \SIa, the attacker can add some fake samples to the original training sequence. In the second case, analysed in Section \ref{sec.SI_CTR_c}, the attacker can replace some of the training samples with fake values (source identification game with replacement of training samples - \SIr). It is worth stressing that, even if the goal of the attacker is to increase the false negative error probability, the training sequence is corrupted regardless of whether $H_0$ or $H_1$ holds, hence, in general, this part of the attack also affects the false positive error probability. As it will be clear later on, this forces the defender to adopt a worst case perspective to ensure that $P_{fp}$ is surely lower than $2^{-{\lambda n}}$.

As to $Y$, we assume that the attacker knows $P_Y$ exactly. For a proper definition of the payoff of the game, we also assume that D knows $P_Y$. This may seem a too strong assumption, however we will show later on that the optimum strategy of D does not depend on $P_Y$, thus allowing us to relax the assumption that D knows $P_Y$.

With the above ideas in mind, we are now ready to give a formal definition of the \SIa ~game.

\subsection{Structure of the  \SIa~game}

%\BTcomm{Titolo: 'Structure of the \SIa ~game'}\\
A schematic representation of the \SIa ~game is given in \fig{fig.ADVsetup_add}.

Let $\tau^{m_1}$ be a sequence drawn from $X$. We assume that $\tau^{m_1}$ is accessible to A, who corrupts it by concatenating to it a sequence of fake samples $\tau^{m_2}$. Then A reorders the overall sequence in a random way so to hide the position of the fake samples. Note that reordering does not alter the statistics of the training sequence since the sequence is supposed to be generated from a memoryless source\footnote{By using the terminology introduced in \cite{Barreno2010}, the above scenario can be referred to as a {\em causative} attack with control over training data.}. In the following, we denote by $m$ the final length of the training sequence ($m = m_1 + m_2$), and by $\alpha = \frac{m_2}{m_1+m_2}$ the portion of fake samples within it. The corrupted training sequence observed by D is indicated by $t^m$. Eventually, we hypothesize a linear relationship between the lengths of the test and the corrupted training sequence, i.e. $m = cn$, for some constant value $c$\footnote{In this paper, we are interested in studying the equilibrium point of the source identification game when the length of the test and training sequences tend to infinity. Strictly speaking, we should ensure that when $n$ grows, all the quantities $m$, $m_1$ and $m_2$ are integer numbers for the given $c$ and $\alpha$. In practice, we will neglect such an issue, since when $n$ grows the ratios $m/n$ and $m_1/(m_1 + m_2)$ can approximate any real values $c$ and $\alpha$. More rigorously, we could consider only rational values of $c$ and $\alpha$, and focus on subsequences of $n$ including only those values for which $m/n = c$ and $m_1/(m_1 + m_2) = \alpha$.}.

The goal of D is to decide if an observed sequence $v^n$ has been drawn from the same source that generated $t^m$ ($H_0$) or not ($H_1$). We assume that D knows that a certain percentage of samples in the training sequence are corrupted, but he has no clue about the position of the corrupted samples. The attacker can also modify the sequence generated by $Y$ so to induce a decision error. The corrupted sequence is indicated by $z^n$. With regard to the two phases of the attack, we assume that A first corrupts the training sequence, then he modifies the sequence $y^n$. This means that, in general, $z^n$ will depend both on $y^n$ and $t^m$, while $t^m$ (noticeably $\tau^{m_2}$) does not depend on $y^n$. Stated in another way, the corruption of the training sequence can be seen as a preparatory part of the attack, whose goal is to ease the subsequent camouflage of $y^n$.

\begin{figure}[t!]
\centering
\includegraphics[width = 0.99\columnwidth]{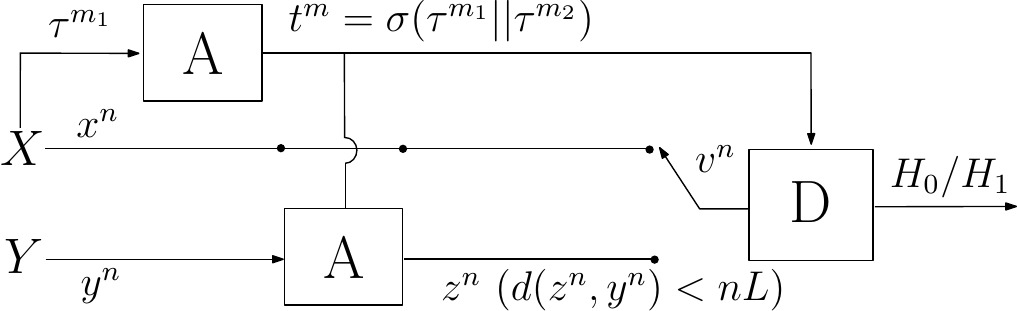}
\caption{Schematic representation of the \SIa~game. Symbol $||$ denotes concatenation of sequences and $\sigma()$ is a random permutation of sequence samples.}
\label{fig.ADVsetup_add}
\end{figure}

For a formal definition of the \SIa~game, we must define the set of strategies available to D and A (respectively $\SS_D$ and $\SS_A$) and the corresponding payoffs.

\subsection{Defender's strategies}

The basic assumption behind the definition of the space of strategies available to D is that to make his decision  D relies only on the first order statistics of $v^n$ and $t^m$. This assumption is equivalent to requiring that the acceptance region for hypothesis $H_0$, hereafter referred to as $\Lambda^{n \times m}$, is a union of pairs of type classes\footnote{We use the superscript $n \times m$ to indicate explicitly that $\Lambda^{n \times m}$ refers to $n$-long test sequences and $(m = cn)$-long training sequences.}, or equivalently, pairs of types $(P,R)$, where $P \in \PP^n$ and $R \in \PP^m$. To define $\Lambda^{n \times m}$, D follows a Neyman-Pearson approach, requiring that the false positive error probability is lower than a certain threshold. Specifically, we require that the false positive error probability tends to zero exponentially fast with a decay rate at least equal to $\lambda$. Given that the pmf $P_X$ ruling the emission of sequences under $H_0$ is not known and given that the corruption of the training sequence is going to impair D's decision under $H_0$, we adopt a worst case approach and require that the constraint on the false positive error probability holds for all possible $P_X$ and for all the possible strategies available to the attacker. Given the above setting, the space of strategies available to D is defined as follows:
\begin{equation}
    \SS_{D} = \{ \Lambda^{n \times m} \subset \PP^n \times \PP^m:~ \max_{P_X \in \mathcal{P}} ~\max_{s \in \SS_A} ~P_{fp} ~\le~ 2^{-\lambda n}\},
\label{eq.SD}
\end{equation}
where the inner maximization is performed over all the strategies available to the attacker. We will refine this definition at the end of the next section, after the exact definition of the space of strategies of the attacker.

\subsection{Attacker's strategies}

With regard to A, the attack consists of two parts. Given a sequence $y^n$ drawn from $P_Y$, and the original training sequence $\tau^{m_1}$, the attacker first generates a sequence of fake samples $\tau^{m_2}$ and mixes them up with those in $\tau^{m_1}$ producing the training sequence $t^m$ observed by D. Then he transforms $y^n$ into $z^n$, eventually trying to generate a pair of sequences ($z^n, t^m$)\footnote{While reordering is essential to hide the position of fake samples to D, it does not have any impact on the position of ($z^n, t^m$) with respect to $\Lambda^{n \times m}$, since we assumed that the defender bases its decision only on the first order statistic of the observed sequences. For this reason, we omit to indicate the reordering operator $\sigma$ in the attacking procedure.} whose types belong to $\Lambda^{n \times m}$. In doing so, he must ensure that $d(y^n, z^n) \le nL$ for some distortion function $d$.

Let us consider the corruption of the training sequence first. Given that the defender bases his decision only on the type of $t^m$, we are only interested in the effect that the addition of the fake samples has on $P_{t^m}$. By considering the different length of $\tau^{m_1}$ and $\tau^{m_2}$, we have:
\begin{equation}
\label{eq.PMF_composition}
    P_{t^m} =  \alpha P_{\tau^{m_2}} + (1-\alpha) P_{\tau^{m_1}},
\end{equation}
where $P_{t^m} \in \PP^m$, $P_{\tau^{m_1}} \in \PP^{m_1}$ and $P_{\tau^{m_2}} \in \PP^{m_2}$. The first part of the attack, then, is equivalent to choosing a pmf in $\PP^{m_2}$ and mixing it up with $P_{\tau^{m_1}}$. By the same token,  it is reasonable to assume that the choice of the attacker depends only on $P_{\tau^{m_1}}$ rather than on the single sequence $\tau^{m_1}$. Arguably, the best choice of the pmf in $\PP^{m_2}$ will depend on $P_Y$, since the corruption of the training sequence is instrumental in letting the defender think that a sequence generated by $Y$ has been drawn by the same source that generated $t^m$.

%\BTcomm{I fear that this is a critical point. In principle the strategy of corruption of the training samples and the one of corruption of the test sequence cannot be decoupled....Con riferimento alla \eqref{eq.optimum_SA}, e' evidente che $Q^*$ dipende anche da $y^n$ non solo dal type $P_{\tau^{m_1}}$ (anche se alla fine questa dipendenza diventa dipendenza da $Y$, credo che dobbiamo essere piu' cauti qui). Per di piu', il modo ottimo di corrompere il traning, cioe' $Q^*$, dipende anche da come viene attaccato il test, cioe',  la (20) non si puo' disaccoppiare (almeno se non si fanno ipotesi sulla misura di distanza $d(,)$....(isotropia?)).
%La mappatura ottima $S_{YZ}^{n,*}$, a sua volta, dipendera' dal valore di $P_{\tau^{m_2}}$ oltre che da $y^n$, anche se di nuovo, sintoticamente, questa dipendenza diventa dipendenza da $P_X$.}

To describe the part of the attack applied to the test sequence, we follow the approach used in \cite{BT_SMargin} based on transportation theory \cite{rachev1998mass}.
Let us indicate by $n(i,j)$ the number of times that the $i$-th symbol of the alphabet is transformed into the $j$-th one as a consequence of the attack. Similarly, let $S^n_{YZ}(i,j) = n(i,j)/n$ be the relative frequency with which such a transformation occurs. In the following, we refer to $S^n_{YZ}$ as {\em transportation map}. For any additive distortion measure, the distortion introduced by the attack can be expressed in terms of $n(i,j)$ and $S^n_{YZ}$. In fact, we have:
\begin{equation}
    d(y^n, z^n) ~=~ \sum_{i,j} n(i,j) d(i,j),
\label{eq.overall_dist}
\end{equation}
\begin{equation}
    \frac{d(y^n, z^n)}{n} ~=~ \sum_{i,j} S^n_{YZ}(i,j) d(i,j).
\label{eq.averagel_dist}
\end{equation}
where $d(i,j)$ is the distortion introduced when symbol $i$ is transformed into symbol $j$.

The map $S^n_{YZ}$ also determines the type of the attacked sequence. In fact, by indicating with $P_{z^n}(j)$ the relative frequency of symbol $j$ into $z^n$, we have:
\begin{equation}
\label{eq.out_type}
    P_{z^n}(j) ~=~ \sum_i S^n_{YZ}(i,j) ~\triangleq~ S^n_Z(j).
\end{equation}
Finally, we observe that the attacker can not change more symbols than there are in the sequence $y^n$; as a consequence a map $S^n_{YZ}$ can be applied to a sequence $y^n$ only if $S^n_{Y}(i) \triangleq \sum_{j} S^n_{YZ}(i,j) = P_{y^n}(i)$. Sometimes, we find convenient to explicit the dependence of the map chosen by the attacker on the type of $t^m$ and $y^n$, and hence we will also adopt the notation $S^n_{YZ}(P_{t^m}, P_{y^n})$.

By remembering that $\Lambda^{n \times m}$ depends on $v^n$ only through its type, and given that the type of the attacked sequence depends on $S^n_Y$ only through $S^n_{YZ}$, we can define the second phase of the attack as the choice of a transportation map among all {\em admissible} maps, a map being admissible if:
\begin{align}
\label{eq.admissiblemap1}
    & S^n_{Y} ~= ~P_{y^n} \\ \nonumber
    & \sum_{i,j} S^n_{YZ}(i,j) d(i,j) ~\le~ L.
\end{align}
Hereafter, we will refer to the set of admissible maps as $\A^n(L, P_{y^n})$.

%In the following, we will refer to the result of such an association as $S^n_{YZ}(y^n)$, or $S^n_{YZ}(i,j;y^n)$, when we need to refer explicitly to the relative frequency with which symbol $i$ is transformed into symbol $j$. In the same way, $S^n_Z(j;y^n)$ indicates the output marginal of $S^n_{YZ}(i,j;y^n)$.

With the above ideas in mind, the set of strategies of the attacker can be defined as follows:
\begin{equation}
    \SS_A ~=~ \SS_{A,T} \times \SS_{A,O},
\label{eq.SAD_cartesian}
\end{equation}
%
%\BTcomm{Mi e' venuto il dubbio che non sia corretto indicare il set come un prodotto cartesiano, visto che le due strategie non sono indipendenti...Per via di questo dubbio, nelle definizioni delle strategie sotto (e in seguito) ho preferito lasciare gli argomenti esplicitati in $Q()$ e $S^n_{YZ}(,)$.} \MB{Dubbio comprensibile ma sbagliato. Lo spazio delle strategie e' un prodotto cartesiano infatti le possibili scelte per la seconda parte dell'attacco non dipendono dalla scelta fatta nella prima fase e quindi lo spazio delle startegie e' un prodotto cartesiano. E' invece vero che la scelta ottima delle due fasi dell'attacco va fatta congiuntamente, ma questo non significa che lo spazio delle strategie non sia un prodotto cartesiano. Per questo motivo ho riscritto tutto come era prima.}
where $\SS_{A,T}$ and $\SS_{A,O}$ indicate, respectively, the part of the attack affecting the training sequence and the observed sequence, and are defined as:
%%
%\begin{align}
%    \SS_{A,T} = \bigg\{ Q(P_{\tau^{m_1}}), \hspace{0.2cm} Q: \PP^{m_1} \rightarrow \PP^{m_2}\bigg\},
%\label{eq.SAD_TR_T}
%\end{align}
%%
%%
%\begin{align}
%    \SS_{A,O} = \bigg\{ S^n_{YZ}(P_{y^n}, P_{t^m}),  \hspace{0.2cm} S^n_{YZ}: \PP^n \times \PP^m \rightarrow \A^n(L, P_{y^n}) \bigg\},
%\label{eq.SAD_TR_O}
%\end{align}
%%
%\MB{A me piacerebbe di piu' cosi':
%\begin{align}
%    \SS_{A,T} = \bigg\{Q(\cdot): \PP^{m_1} \rightarrow \PP^{m_2}\bigg\},
%\label{eq.SAD_TR_T_MB}
%\end{align}
%%
%%
%\begin{align}
%    \SS_{A,O} = \bigg\{ S^n_{YZ}(\cdot,\cdot): \PP^n \times \PP^m \rightarrow \A^n(L, P_{y^n}) \bigg\},
%\label{eq.SAD_TR_O_MB}
%\end{align}
%%
%}
%
\begin{align}
    \SS_{A,T} & ~=~ \bigg\{ Q(P_{\tau^{m_1}}): ~\PP^{m_1} \rightarrow \PP^{m_2}\bigg\},
\label{eq.SAD_TR_T}\\
    \SS_{A,O} & ~=~ \bigg\{ S^n_{YZ}(P_{y^n}, P_{t^m}): ~\PP^{n} \times \PP^{m}  \rightarrow \A^n(L, P_{y^n}) \bigg\}.
\label{eq.SAD_TR_O}
\end{align}
%

%
%\begin{align}
%& \SS_{A,T} = \bigg\{ Q(P_{\tau^{m_1}}), \hspace{0.2cm} Q: \PP^{m_1} \rightarrow \PP^{m_2}\bigg\},
%\label{eq.SAD_TR_T} \\
%   &  \SS_{A,O} = \bigg\{ S^n_{YZ}(P_{y^n}, P_{t^m}) \in \A^n(L, P_{y^n}), \nonumber\\
%    & \hspace{3.5cm} S^n_{YZ}: \PP^n \times \PP^m \rightarrow \PP^{2n} \bigg\},
%\label{eq.SAD_TR_O}
%\end{align}
%
%\BT{where $P_{t^m} = (1- \alpha) P_{\tau^{m_1}} + \alpha Q(P_{\tau^{m_1}})$.}\\ Abbiamo gia' definito $P_{t^m}$ perche' lo vuoi ridefinire ?
\noindent Note that the first part of the attack ($\SS_{A,T}$) is applied regardless of whether $H_0$ or $H_1$ holds, while the second part ($\SS_{A,O}$) is applied only under $H_1$. We also stress that the choice of $Q(P_{\tau^{m_1}})$ depends only on the training sequence $\tau^{m_1}$,  while the transportation map used in the second phase of the attack is a function of both on $y^n$ and $\tau^{m_1}$ (through $t^m$).
Finally, we observe that with these definitions, the set of strategies of the defender can be redefined by explicitly indicating that the constraint on the false positive error probability must be verified for all possible choices of $Q(\cdot) \in \SS_{A,T}$, since this is the only part of the attack affecting $P_{fp}$. Specifically, we can rewrite \eqref{eq.SD} as
\begin{equation}
    \SS_{D} = \{ \Lambda^{n \times m}  \subset  \PP^n \times  \PP^m: ~\max_{P_X} ~\max_{Q(\cdot) \in \SS_{A,T}} P_{fp} ~\le~ 2^{-\lambda n}\}.
\label{eq.SD_bis}
\end{equation}

\subsection{Payoff}

The payoff is defined in terms of the false negative error probability, namely:
\begin{equation}
    u(\Lambda^{n \times m}, (Q(\cdot), ~S^n_{YZ}(\cdot, \cdot))) ~=~ -P_{fn}.
\label{eq.payoff_TR}
\end{equation}
Of course, D aims at maximising $u$ while A wants to minimise it.

\subsection{The \SIa~game with targeted corruption (\SIat~game)}

The \SIa~game is difficult to solve directly, because of the 2-step attacking strategy.
We will work around this difficulty by tackling first with a slightly different version of the game, namely the source identification game with targeted corruption of the training sequence, \SIat, depicted in Fig. \ref{fig.ADVsetup_add_bis}.

Whereas the strategies available to the defender remain the same, for the attacker, the choice of $Q(\cdot)$ is targeted to the counterfeiting of a given sequence $y^n$. In other words, we will assume that the attacker corrupts the training sequence $\tau^{m_1}$ to ease the counterfeiting of a specific sequence $y^n$ rather than to increase the probability that the second part of the attack succeeds. This means that the part of the attack aiming at corrupting the training sequence also depend on $y^n$, that is:
\begin{align}
    \SS_{A,T} ~=~ \bigg\{Q(P_{\tau^{m_1}}, P_{y^n}): ~\PP^{m_1} \times \PP^n \rightarrow \PP^{m_2}\bigg\}.
\label{eq.SAD_TR_Targeted_MB}
\end{align}
Even if this setup is not very realistic and is more favourable to the attacker, who can exploit the exact knowledge of $y^n$ (rather than its statistical properties) also for the corruption of the training sequence, in the next section we will show that, for large $n$, the \SIat~game is equivalent to the non-targeted version of the game we are interested in.

With the above ideas in mind, the \SIat~game is formally defined as follows.
\begin{figure}[t!]
\centering \includegraphics[width = 0.99\columnwidth]{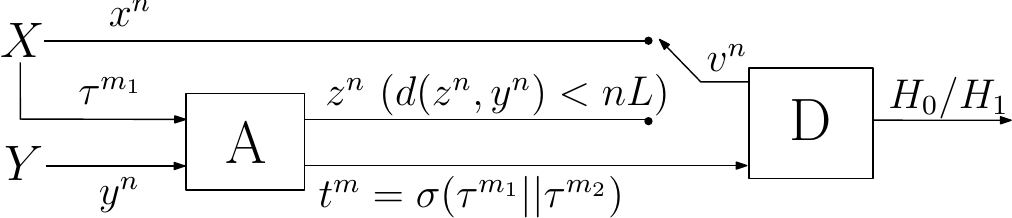}
\caption{\SIa~game with targeted corruption of the training sequence.}
\label{fig.ADVsetup_add_bis}
\end{figure}

\subsubsection{Defender's strategies}

\begin{equation}
    \SS_{D} = \{ \Lambda^{n \times m} \subset \PP^n \times  \PP^m: ~ \max_{P_X} \max_{Q(\cdot,\cdot) \in \SS_{A,T}} P_{fp} \le 2^{-\lambda n}\}.
\label{eq.SD_bis_t}
\end{equation}

\subsubsection{Attacker's strategies}

\begin{equation}
    \SS_A ~=~ \SS_{A,T} \times \SS_{A,O}
\label{eq.SAD_cartesian_Tt}
\end{equation}
with $\SS_{A,T}$ and $\SS_{A,O}$ defined as in \eqref{eq.SAD_TR_Targeted_MB} and \eqref{eq.SAD_TR_O} respectively.

\subsubsection{Payoff}

The payoff is still equal to the false negative error probability:
\begin{equation}
   % u(\Lambda^{n \times m}, (Q(P_{\tau^{m_1}}, P_{y^n}), S^n_{YZ}(P_{y^n}, P_{t^m}))) = -P_{fn},
    u(\Lambda^{n \times m}, (Q(\cdot, \cdot), ~S^n_{YZ}(\cdot, \cdot))) ~=~ -P_{fn}.
\label{eq.payoff_TR_t}
\end{equation}
%

%\section{Asymptotic equilibrium, payoff and source distinguishability for the \SIa~game}
\section{Asymptotic equilibrium and payoff of the \SIat~and \SIa~games}
\label{sec.SI_CTR_add_solution}

In this section, we derive the asymptotic equilibrium point of the  \SIat~and the \SIa~games when the length of the test and training sequences tends to infinity and evaluate the payoff at the equilibrium.

\subsection{Optimum defender's strategy}

We start by deriving the asymptotically optimum strategy for D. As we will see, a dominant and universal strategy with respect to $P_Y$ exists for D. In other words, the optimum choice of D depends on neither the strategy chosen by the attacker nor $P_Y$. In addition, since the constraint on the false positive probability must be satisfied for all attackers' strategy, the optimum strategy for the defender is the same for both the targeted and non-targeted versions of the game.

As a first step, we look for an explicit expression of the false positive error probability. Such a probability depends on $P_X$ and on the strategy used by A to corrupt the training sequence. In fact, the mapping of $y^n$ into $z^n$ does not have any impact on D's decision under $H_0$.
We carry out our derivations by focusing on the game with targeted corruption. It will be clear from our analysis that the dependence on $y^n$ has no impact on $P_{fp}$, and hence the same results hold for the game with non-targeted corruption.

For a given $P_X$ and $Q(\cdot, \cdot)$, $P_{fp}$ is equal to the probability that $Y$ generates a sequence $y^n$ and $X$ generates two sequences $x^n$ and $\tau^{m_1}$, such that the pair of type classes $(P_{x^n}, \alpha Q(P_{\tau^{m_1}}, P_{y^n}) + (1-\alpha) P_{\tau^{m_1}})$ falls outside $\Lambda^{n \times m}$. Such a probability can be expressed as:
\begin{align}
\label{eq.explicit_Pfp}
P_{fp} ~=~ & Pr\{(P_{x^n}, \alpha Q(P_{\tau^{m_1}}, P_{y^n}) + (1-\alpha) P_{\tau^{m_1}}) \in \bar{\Lambda}^{n\times m} \} \nonumber \\
    ~=~ & \sum_{P_{y^n} \in \PP^n} P_Y(T(P_{y^n}))  \cdot \\
    & \sum_{(P_{x^n}, P_{t^m}) \in \bar{\Lambda}^{n\times m}} \hspace{-0.5cm} P_X(T(P_{x^n})) \cdot \hspace{-1.7cm} \sum_{\substack{P_{\tau^{m_1}} \in \PP^{m_1}:  \\ \alpha Q(P_{\tau^{m_1}}, P_{y^n}) + (1-\alpha) P_{\tau^{m_1}} = P_{t^m}}}\hspace{-1.6cm}  P_X(T(P_{\tau^{m_1}})),\nonumber
\end{align}
where $\bar{\Lambda}^{n \times m}$ is the complement of $\Lambda^{n \times m}$, and where we have exploited the fact that under $H_0$ the training sequence $\tau^{m_1}$ and the test sequence $x^n$ are generated independently by $X$. Given the above formulation, the set of strategies available to D can be rewritten as:
\begin{align}
\label{eq.SD_explicit}
\SS_D & = \bigg \{ \Lambda^{n \times m} : ~ \max_{P_X} \max_{Q(\cdot, \cdot)}  \sum_{P_{y^n} \in \PP^n} P_Y(T(P_{y^n})) \cdot  \\ \nonumber
& \sum_{(P_{x^n}, P_{t^m}) \in \bar{\Lambda}^{n\times m}} \hspace{-0.6cm} P_X(T(P_{x^n})) \cdot \hspace{-1.5cm} \sum_{\substack{P_{\tau^{m_1}} \in \PP^{m_1}:  \\ \alpha Q(P_{\tau^{m_1}}, P_{y^n}) + (1-\alpha) P_{\tau^{m_1}} = P_{t^m}}}\hspace{-1.6cm}  P_X(T(P_{\tau^{m_1}})) \le 2^{-\lambda n}
\bigg \}.
\end{align}
%
%Our results will be expressed in terms of the generalized log-likelihood ratio function $h(P_{x^n}, P_{t^m})$ defined as (see \cite{BTtit,Gut89,Kendall}):
%%
%\begin{equation}
%    h(P_{x^n}, P_{t^m}) = \DD(P_{x^n} || P_{r^{n + m}}) + \frac{m}{n} \DD(P_{t^m} || P_{r^{n+ m}}),
%\label{eq.h}
%\end{equation}
%%
%where $P_{r^{n + m}}$ denotes the type of the sequence $r^{n+m}$, obtained by concatenating $x^n$ and $t^m$, i.e. $r^{n+m} = x^n \| t^m$. The intuitive meaning behind the above definition is that $P_{r^{n+m}}$ is the pmf which maximises the probability that a memoryless source generates two independent sequences belonging to $T(P_{x^n})$ and $T(P_{t^m})$, and that such a probability is equal to $2^{-n h(P_{x^n}, P_{t^m})}$ at the first order in the exponent (see \cite{Kendall} or Lemma 1 in \cite{BTtit}).
%
%With the above definitions and notations, we

We are now ready to prove the following lemma, which describes the asymptotically optimum strategy for the defender for both versions of the game.

\begin{lemma}
Let $\Lambda^{n \times m,*}$ be defined as follows:
\begin{equation}
    \Lambda^{n \times m,*} = \left\lbrace  (P_{v^n}, P_{t^m}) : \hspace{-0.06cm}
    \min_{Q \in \PP^{m_2}} \hspace{-0.06cm}  h\left(P_{v^n},\frac{P_{t^m}-\alpha Q}{1-\alpha} \right) ~\le~ \lambda \hspace{-0.03cm} - \hspace{-0.03cm} \delta_{n} \right\rbrace
\label{eq.optimum_SD}
\end{equation}
with
\begin{equation}
\delta_{n} = |\XX|\frac{\log(n+1)((1-\alpha)nc+1)}{n},
\end{equation}
where $|\XX|$ is the cardinality of the source alphabet and where the minimisation over $Q$ is limited to all the $Q$'s such that $P_{t^m}-\alpha Q$ is nonnegative for all the symbols in $\XX$.

\noindent Then:
\begin{enumerate}
\item{$ \max\limits_{P_X} \max\limits_{s \in \SS_A}  ~P_{fp} ~\le~ 2^{-n(\lambda - \nu_n)}$, with $\lim\limits_{n \rarrow \infty} \nu_n = 0$,}
\item{$\forall \Lambda^{n \times m} ~\in~ \SS_{D}$, we have $\bar{\Lambda}^{n \times m} \subseteq {\bar{\Lambda}^{n \times m,*}}$}.
\end{enumerate}
\label{lemma.optimum_SD}
\end{lemma}
\begin{proof}
To prove the first part of the lemma, we see that from the expression of the false positive error probability given by eq. (\ref{eq.explicit_Pfp}), we can write:
\begin{align}
\label{eq.proof_L1_first_part}
& \max_{P_X} ~\max_{Q(\cdot,\cdot)} ~P_{fp} ~\le~ \\ \nonumber
& \max_{P_X}  \sum_{P_{y^n} \in \PP^n} P_Y(T(P_{y^n})) \cdot \sum_{\substack{(P_{x^n}, P_{t^m}) \\ \in \bar{\Lambda}^{n \times m,*}}} \hspace{-0.3cm} P_X(T(P_{x^n})) \cdot \nonumber\\
& \hspace{3cm}  \max_{Q(\cdot, \cdot)} \hspace{-1cm} \sum_{\substack{P_{\tau^{m_1}} \in \PP^{m_1}:  \\ \alpha Q(P_{\tau^{m_1}}, P_{y^n}) + (1-\alpha) P_{\tau^{m_1}} = P_{t^m}}}\hspace{-1.6cm}  P_X(T(P_{\tau^{m_1}})).
\end{align}
Let us consider the term within the inner summation. For each $P_{\tau^{m_1}}$ such that $\alpha Q(P_{\tau^{m_1}}, P_{y^n}) + (1-\alpha) P_{\tau^{m_1}} = P_{t^m}$, we have\footnote{It is easy to see that the bound \eqref{eq.innersum} holds also for the non-targeted game, when $Q$ depends on the training sequence only ($Q(P_{\tau^{m_1}})$).}:
\begin{equation}
    P_X(T(P_{\tau^{m_1}})) ~\le~ \max_{Q \in \PP^{m_2}} P_X\left( T \left( \frac{P_{t^m} - \alpha Q}{1 - \alpha} \right) \right),
\label{eq.innersum}
\end{equation}
with the understanding that the maximisation is carried out only over the $Q$'s such that $P_{t^m} - \alpha Q$ is nonnegative for all the symbols in $\XX$.

Thanks to the above observation, we can upper bound the false positive error probability as follows:
\begin{align}
\label{eq.Pf_upbound}
& \max_{P_X} ~\max_{Q(\cdot, \cdot)} ~P_{fp} ~\le \\ \nonumber
& \max_{P_X}  \sum_{P_{y^n} \in \PP^n} P_Y(T(P_{y^n})) \cdot \\ \nonumber
& \sum_{\substack{(P_{x^n}, P_{t^m}) \\ \in \bar{\Lambda}^{n \times m,*}}} \hspace{-0.4cm} P_X(T(P_{x^n})) \cdot |\PP^{m_1}|  \cdot \max_{Q \in \PP^{m_2}}  P_X \left( T \left( \frac{P_{t^m} - \alpha Q}{1 - \alpha} \right) \right) \\ \nonumber
&  \stackrel{(a)}{=}  \max_{P_X}  \sum_{\substack{(P_{x^n}, P_{t^m}) \\ \in \bar{\Lambda}^{n \times m,*}}} \hspace{-0.2cm} P_X(T(P_{x^n})) |\PP^{m_1}|  \max_{Q \in \PP^{m_2}} P_X \left( T \left( \frac{P_{t^m} - \alpha Q}{1 - \alpha} \right) \hspace{-0.1cm} \right) \\ \nonumber
&  \le |\PP^{m_1}|  \hspace{-0.1cm}\sum_{\substack{(P_{x^n}, P_{t^m}) \\ \in \bar{\Lambda}^{n \times m,*}}} \hspace{-0.1cm} \max_{Q \in \PP^{m_2}}  \max_{P_X} \hspace{0.1cm} P_X(T(P_{x^n})) P_X \left( T \left( \frac{P_{t^m} - \alpha Q}{1 - \alpha} \right) \hspace{-0.1cm} \right)
\end{align}
where in $(a)$ we exploited the fact that the rest of the expression no longer depends on $P_{y^n}$.
From this point, the proof goes along the same line of the proof of Lemma 2 in \cite{BTtit}, by observing that $\max_{P_X} P_X(T(P_{x^n})) P_X \left( T \left( \frac{P_{t^m} - \alpha Q}{1 - \alpha} \right) \right)$is upper bounded by $2^{-n h(P_{x^n}, \frac{P_{t^m} - \alpha Q}{1 - \alpha})}$, and that for each pair of types in $\bar{\Lambda}^{n\times m, *}$, $h(P_{x^n}, \frac{P_{t^m} - \alpha Q}{1 - \alpha})$ is larger than $\lambda - \delta_n$ for every $Q$ by the very definition of $\Lambda^{n\times m, *}$.

We now pass to the second part of the lemma. Let $\Lambda^{n \times m}$ be a strategy in $\SS_D$, and let $(P_{x^n}, P_{t^m})$ be a pair of types contained in $\bar{\Lambda}^{n \times m}$. Given that $\Lambda^{n \times m}$ is an admissible decision region (see \eqref{eq.SD_bis_t}), the probability that $X$ emits a test sequence belonging to $T(P_{x^n})$ and a training sequence $\tau^{m_1}$ such that after the attack $(\tau^{m_1} || \tau^{m_2}) \in T(P_{t^m})$ must be lower than $2^{-\lambda n}$ for all $P_X$ and all possible attacking strategies, that is:
\begin{align}
\label{eq.proof_L1_second_part}
2^{-\lambda n} & ~> ~\max_{P_X} \max_{Q(\cdot, \cdot)}  \sum_{P_{y^n} \in \PP^n} P_Y(T(P_{y^n})) \cdot \\ \nonumber
& \hspace{1.2cm} \big[ P_X(T(P_{x^n})) ~ \cdot  \hspace{-1.5cm}\sum_{\substack{P_{\tau^{m_1}} :  \\ \alpha Q(P_{\tau^{m_1}}, P_{y^n}) + (1-\alpha) P_{\tau^{m_1}} = P_{t^m}}} \hspace{-1.5cm}P_X(T(P_{\tau^{m_1}}))\big] \\ \nonumber  \\ \nonumber
& \stackrel{(a)}{=} ~ \max_{P_X}   \sum_{P_{y^n} \in \PP^n} P_Y(T(P_{y^n})) \cdot \\ \nonumber
& \hspace{1.2cm} \big[ P_X(T(P_{x^n})) \cdot \max_{Q(\cdot, P_{y^n})}  \hspace{-1.5cm} \sum_{\substack{P_{\tau^{m_1}} :  \\ \alpha Q(P_{\tau^{m_1}}, P_{y^n}) + (1-\alpha) P_{\tau^{m_1}} = P_{t^m}}} \hspace{-1.5cm}P_X(T(P_{\tau^{m_1}})) \big] \\ \nonumber \\ \nonumber
& \stackrel{(b)}{\ge} ~ \max_{P_X}   \sum_{P_{y^n} \in \PP^n} P_Y(T(P_{y^n})) \cdot \big[ P_X(T(P_{x^n})) \cdot \\ \nonumber
& \hspace{1.2cm} \max_{Q(P_{\tau^{m_1}}, P_{y^n})} P_X \left( T \left( \frac{P_{t^m}-\alpha Q(P_{\tau^{m_1}}, P_{y^n})}{1-\alpha} \right) \right) \bigg] \\ \nonumber
& \stackrel{(c)}{=} ~ \max_{P_X} ~P_X(T(P_{x^n}))  \max_{Q \in \PP^{m_2}} P_X  \left( T \left( \frac{P_{t^m}-\alpha Q}{1-\alpha} \right)  \right),
\end{align}
where ${(a)}$ is obtained by replacing the maximisation over all possible strategies $Q(\cdot,\cdot)$, with a maximisation over $Q(\cdot,P_{y^n})$ for each specific $P_{y^n}$, and $(b)$ is obtained by considering only one term $P_{\tau^{m_1}}$ of the inner summation and optimising $Q(P_{\tau^{m_1}},P_{y^n})$ for that term. Finally, $(c)$ follows by observing that the optimum $Q(\cdot,P_{y^n})$ is the same for any $P_{y^n}$. As usual, the maximization over $Q$ in the last expression is restricted to the $Q$'s for which $P_{t^m}-\alpha Q~\ge~ 0$ for all the symbols in $\XX$ \footnote{It is easy to see that the same lower bound can be derived also for the non targeted case, as the optimum $Q$ in the second to last expression does not depend on $P_{y^n}$.}

By lower bounding the probability that a memoryless source $X$ generates a sequence belonging to a certain type class (see \cite{CandT}, chapter 12), we can continue the above chain of inequalities as follows
\begin{align}
\label{eq.Pf_lowerbound}
2^{-\lambda n} & ~>~ \frac{\max\limits_{P_X} \max\limits_{Q \in \PP^{m_2}} 2^{-n\big[\DD(P_{x^n}||P_X)+\frac{m_1}{n}\DD\big( \frac{P_{t^m}-\alpha Q}{1-\alpha}|| P_X\big)\big]}  }{(n+1)^{|\XX|}(m_1+1)^{|\XX|}} \\ \nonumber
& ~\ge~ \frac{2^{-n \min\limits_{Q \in \PP^{m_2}} \min\limits_{P_X}  \big[\DD(P_{x^n}||P_X)+\frac{m_1}{n}\DD\big( \frac{P_{t^m}-\alpha Q}{1-\alpha}|| P_X\big)\big]   }}{(n+1)^{|\XX|}(m_1+1)^{|\XX|}} \\ \nonumber
& ~\stackrel{(a)}{=}~ \frac{2^{-n \min\limits_{Q \in \PP^{m_2}}  h\big( P_{x^n} , \frac{P_{t^m} - \alpha Q}{1 - \alpha}  \big)  }}{(n+1)^{|\XX|}(m_1+1)^{|\XX|}},
\end{align}
where $(a)$ derives from the minimization properties of the generalised log-likelihood ratio function $h()$ (see Lemma 1, in \cite{BTtit}). By taking the $\log$ of both terms we have:
\begin{equation}
%    \min_{Q \in \PP^{m_2}} h\left( P_{x^n} , \frac{P_{t^m} - \alpha Q}{1 - \alpha}  \right) \ge \lambda - \frac{|\XX| \log_2(n+1)(m_1+1)}{n},
    \min_{Q \in \PP^{m_2}} h\left( P_{x^n} , \frac{P_{t^m} - \alpha Q}{1 - \alpha}  \right) ~>~ \lambda - \delta_n,
\label{eq.final_expr}
\end{equation}
thus completing the proof of the lemma.
\end{proof}

Lemma 1 shows that the strategy $\Lambda^{n \times m,*}$ is asymptotically admissible (point 1) and optimal (point 2),  regardless of the attack. From a game-theoretic perspective, this means that such a strategy is a dominant strategy for D and  implies that the game is dominance solvable \cite{ChenGames}. Similarly, the optimum strategy is a semi-universal one, since it depends on $P_X$ but it does not depend on $P_Y$.

It is clear from the proof of Lemma 1 that the same optimum strategy holds for the targeted and non-targeted versions of the game. The situation is rather different with regard to the optimum strategy for the attacker. Despite the existence of a dominant strategy for the defender, in fact, the identification of the optimum attacker's strategy for the \SIa~game is not easy due to the 2-step nature of the attack.
For this reason, in the following sections, we will focus on the targeted version of the game, which is easier to study. We will then use the results obtained for the \SIat~game to derive the best achievable performance for the case of non-targeted attack.

\subsection{The \SIat~game: optimum attacker's strategy and equilibrium point}
\label{sec.payEQ}

Given the dominant strategy of D, for any given $\tau^{m_1}$ and $y^n$, the optimum attacker's strategy for the \SIat~game boils down to the following double minimisation:
\begin{align}
\label{eq.optimum_SA_double}
    & (Q^*(P_{\tau^{m_1}}, P_{y^n}), ~S^{n,*}_{YZ}(P_{y_n}, P_{t^m})) ~= \\ \nonumber
    & \arg\hspace{-0.2cm}\min\limits_{\substack{Q \in \PP^{m_2} \\ S^n_{YZ} \in \A^n(L, P_{y^n})}} \left(  \min_{Q'}  ~h\left( P_{z^n} , \frac{(1-\alpha)P_{\tau^{m_1}} + \alpha Q - \alpha Q'}{1 - \alpha}  \right) \right),
\end{align}
where $P_{z^n}$ is obtained by applying the transformation map $S^n_{YZ}$ to $P_{y^n}$, and where $P_{t^m} = (1-\alpha)P_{\tau^{m_1}} + \alpha Q$. As usual, the minimisation over $Q'$ is limited to the $Q'$ such that all the entries of the resulting pmf are nonnegative.

%Even if this setup is not very realistic and is more favourable to the attacker which can exploit the exact knowledge of $y^n$ rather than its statistical properties only \BT{even for the corruption of the training}, in the next section we will show that when $n \rarrow \infty$ the optimum attacker's strategy depends only on $P_Y$, hence proving that, at least asymptotically,\BT{(...at least for large $n$,)} the \SIa~game with targeted corruption depicted in Fig. \ref{fig.ADVsetup_add_bis} is equivalent to the non-targeted version of the game we are interested in.
%
%\BTcomm{Forse dobbiamo specificare meglio come mai l'approccio del caso targeted non e' realistico, cioe' sottolineare che con la definizione accoppiata la sequenza di training e' costruita per una certa sequenza e cambia quindi con la sequenza di test!!!}

As a remark, for $L = 0$ (corruption of the training sequence only), we get:
\begin{align}
 Q^*(& P_{\tau^{m_1}},P_{y^n})  ~= \nonumber \\
 &\arg  \min\limits_{Q \in \PP^{m_2}} \left[ \min_{Q'}  ~h\left( P_{y^n} , ~P_{\tau^{m_1}} + \frac{\alpha}{1 - \alpha}(Q - Q')  \right) \right],
\end{align}
while, for $\alpha = 0$ (classical setup, without corruption of the training sequence) we have:
\begin{align}
S^{n,*}_{YZ}(P_{y^n}, P_{t^m}) =
    \underset{ S^n_{YZ} \in \A^n(L, P_{y^n})}{\arg \min}  h(P_{z^n}, P_{t^m}),
\end{align}
falling back to the known case of source identification with uncorrupted training, already studied in \cite{BTtit}.
Having determined the optimum strategies of both players, it is immediate to state the following:
\begin{theorem}
The \SIat~game is a dominance solvable game, whose only rationalizable equilibrium corresponds to the profile %$(\Lambda^{n \times m,*}, (Q^*(P_{\tau^{m_1}}, P_{y^n}), ~S^{n,*}_{YZ}(P_{y^n}, P_{t_m}))$.
$(\Lambda^{n \times m,*}, (Q^*(P\cdot, \cdot), ~S^{n,*}_{YZ}(\cdot, \cdot))$.
\label{teo.aseq}
\end{theorem}
\begin{proof}
The theorem is a direct consequence of the fact that $\Lambda^{n \times m,*}$ is a dominant strategy for D.
\end{proof}
We remind that the concept of rationalizable equilibrium is  much stronger than the usual notion of Nash equilibrium, since the strategies corresponding to such an equilibrium are the only ones that two rational players may adopt \cite{Osb94,ChenGames}.

\subsection{The \SIat~game: payoff at the equilibrium}
\label{subsec.SIapayoff}

In this section, we derive the asymptotic value of the payoff at the equilibrium, to see who and under which conditions is going to {\em win} the game.

To start with, we identify the set of pairs $(P_{y^n},P_{\tau^{m_1}})$ for which, as a consequence of A's action, D accepts $H_0$:
\begin{align}
\label{gamma_n}
\hspace{-0.7cm}\Gamma^n(\lambda,\alpha, L) ~=~ \{& ( P_{y^n}, P_{\tau^{m_1}}) : ~ \exists~ (P_{z^n}, P_{t^m}) \in \Lambda^{n \times m,*}  \\ &  \text{s.t. }   P_{t^m} = (1 - \alpha) P_{\tau^{m_1}} + \alpha Q ~\text{and}~ P_{z^n} = S_Z^n \nonumber \\ & \text{for some } Q \in \PP^{m_2} \text{ and } S_{YZ}^n \in \A(L, P_{y^n})  \}.\nonumber
\end{align}
If we fix the type of the non-corrupted training sequence ($P_{\tau^{m_1}}$), we obtain:
\begin{align}
\label{set_Gamma_n_L_1}
\Gamma^n( P_{\tau^{m_1}},\lambda,\alpha,L)  = \{ & P_{y^n}: ~\exists~ P_{z^n} \in \Lambda^{n, *}((1 - \alpha) P_{\tau^{m_1}} + \alpha Q)\\
& \text{s.t. } P_{z^n} = S_Z^n \nonumber \\
& \text{for some }  Q \in \PP^{m_2} \text{ and }  S_{YZ}^n \in \A(L, P_{y^n})  \}, \nonumber
\end{align}
where $\Lambda^{n, *}(P)$ denotes the acceptance region for a fixed type of the training sequence in $\PP^m$.
It is interesting to notice that, since in the current setting A has two degrees of freedom, the attack has a twofold effect: the sequence $y^n$ is modified in order to bring it inside the acceptance region $\Lambda^{n,*}(P_{t^m})$ and the acceptance region itself is modified so to facilitate the former action.

\noindent To go on, we find it convenient to rewrite the set $\Gamma^n(P_{\tau^{m_1}},\lambda,\alpha,L)$ as follows:
\begin{align}
\label{set_Gamma_n_L_1_2}
& \Gamma^n( P_{\tau^{m_1}}, \lambda,\alpha,L) ~= \\
& \{ P_{y^n}: ~\exists S_{PV}^n ~\in~ \A(L, P_{y^n}) \text{ s.t. } S_V^n ~\in~ \Gamma^n_0(P_{\tau^{m_1}}, \lambda, \alpha)\},\nonumber
\end{align}
where
\begin{align}
\label{set_Gamma_n_L_1_2_0}
 & \Gamma^n_0(P_{\tau^{m_1}}, \lambda,\alpha) =  \\
 & \left\{P_{y^n}:  ~\exists Q ~\in~ \PP^{m_2} \text{ s.t. } P_{y^n} \in \Lambda^{n,*}((1 - \alpha) P_{\tau^{m_1}} + \alpha  Q) \right\},\nonumber
\end{align}
is the set containing all the test sequences  (or, equivalently, test types) for which it is possible to corrupt the training set in such a way that they fall within the acceptance region. As the subscript $0$ suggests, this set corresponds to the set in \eqref{set_Gamma_n_L_1} when A cannot modify the sequence drawn from $Y$ (i.e. $L = 0$) and then tries to hamper the decision by corrupting the training sequence only.

By considering the expression of the acceptance region, the set  $\Gamma^n_0(P_{\tau^{m_1}}, \lambda, \alpha)$ can be expressed in a more explicit form as follows:
\begin{align}
\label{definition_Gamma}
 \Gamma^n_0(& P_{\tau^{m_1}},\lambda,\alpha) ~=~ \big\{P_{y^n}: ~\exists Q,Q' ~\in~ \PP^{m_2} \text{ s.t. } \\
& \hspace{0.2cm} h\bigg(P_{y^n}, P_{\tau^{m_1}} +  \frac{\alpha}{(1 - \alpha)} (Q - Q')\bigg) ~\le~ \lambda - \delta_{n} \big\},\nonumber
\end{align}
where the second argument of $h()$ denotes a type in $\PP^{m_1}$ obtained from the original training sequence $\tau^{m_1}$ by first adding $m_2$ samples and later removing (in a possibly different way) the same number of samples. Note that in this formulation $Q$ accounts for the fake samples introduced by the attacker and $Q'$ for the worst case {\em guess} made by the defender of the position of the corrupted samples. We also observe that since we are treating the \SIat~game, in general $Q$ will depend on $P_{y^n}$.
As usual, we implicitly assume that $Q$ and $Q'$ are chosen in such a way that $P_{\tau^{m_1}} +  \frac{\alpha}{(1 - \alpha)} (Q - Q')$ is nonnegative and smaller than or equal to 1 for all the alphabet symbols.

We are now ready to derive the asymptotic payoff of the game by following a path similar to that used in \cite{BT13}, \cite{BTtit}.
First of all we generalise the definition of the sets $\Lambda^{n \times m,*}$, $\Gamma^n$ and $\Gamma^n_0$ so that they can be evaluated for a generic pmf in $\PP$ (that is, without requiring that the pmf's are induced by sequences of finite length). This step passes through the generalization of the $h$ function. Specifically, given any pair of pmf's $(P,P') \in \PP \times \PP$, we define:
\begin{align}
\label{eq.hc}
    & h_c(P,P') ~= ~\DD(P||U) ~+~ c \DD(P'||U); \\ \nonumber
    & U  ~=~ \frac{1}{1+c}P ~+~ \frac{c}{1+c}P',
\end{align}
where $c \in [0,1]$. Note that when $(P,P') \in \PP^n \times \PP^n$, $h_c(P,P') = h(P,P').$
The asymptotic version of $\Lambda^{n\times m,*}$ is:
\begin{equation}
\Lambda^* = \left\lbrace  (P, R) ~: ~\min_{Q}  ~h_c \left(P, ~\frac{R - \alpha Q}{1-\alpha} \right) ~\le~ \lambda \right\rbrace.
\label{eq.asymptotic_lambda}
\end{equation}
%
%In a similar way, the asymptotic versions of  $\Gamma^n$ and $\Gamma^n_0$ are derived from \eqref{set_Gamma_n_L_1_2} and \eqref{set_Gamma_n_L_1_2_0}-\eqref{definition_Gamma} as follows:
In a similar way, we can derive the asymptotic versions of  $\Gamma^n$ and $\Gamma^n_0$ in \eqref{set_Gamma_n_L_1_2} and \eqref{set_Gamma_n_L_1_2_0}-\eqref{definition_Gamma}. To do so,
we first observe that, the transportation map $S_{YZ}^n$ depends on the sources only through the pmfs. By denoting with $S_{PV}^n$ a transportation map from a pmf $P \in \PP^n$ to another pmf $V \in \PP^n$ and rewriting the set $\Gamma^n$ accordingly, we can easily derive the asymptotic version of the set as follows:
\begin{align}
\label{newGamma}
\Gamma(R, \lambda,\alpha,L) ~=~  \{P \in \PP: ~\exists S_{PV} \in \A(L, P) \text{ s.t. } V \in \Gamma_0(R, \lambda, \alpha)\},
\end{align}
with
\begin{align}
\label{newGamma0}
& \Gamma_0(R,\lambda,\alpha) ~=\\
&  \left\{P \in \PP: ~\exists Q \in \PP \text{ s.t. } P \in \Lambda^{*}((1 - \alpha) R + \alpha Q) \right\} ~= \nonumber\\
& \bigg\{P \in  \PP: ~\exists Q,Q' \in \PP \text{ s.t. } h_c\left(P, ~R +  \frac{\alpha}{(1 - \alpha)} (Q - Q')\right) ~\le~ \lambda \bigg\},\nonumber
\end{align}
where the definitions of $S_{PV}$ and $\mathcal{A}(L,P)$ derive from those of $S_{PV}^n$ and $\mathcal{A}^n(L,P)$ by relaxing the requirement that the terms $S_{PV}(i,j)$ and $P(i)$ are rational number with denominator $n$.
%In the following we will refer to such sets as $\Gamma(R,\lambda,\alpha,L)$, $\Gamma_0(R,\lambda,\alpha)$ and $\Lambda^{*}$.
We now have all the necessary tools to prove the following theorem.

\begin{theorem}[Asymptotic payoff of the \SIat~ game]
\label{theo_as_payoff_si}
For the \SIat~ game, the false negative error exponent at the equilibrium is given by
\begin{equation}
    \varepsilon ~=~ \min_{R} [ (1-\alpha)c \DD(R || P_X) + \min_{P \in \Gamma(R, \lambda, \alpha, L)} \DD (P || P_Y)].
\label{eq.fnerr_exp_L}
\end{equation}
Accordingly,
 \begin{enumerate}\label{e.e.cases_tr}
    \item if $P_Y ~\in~ \Gamma(P_X, \lambda, \alpha, L)$ \quad then \quad $\varepsilon ~=~ 0$;
    \item if $P_Y ~\notin~ \Gamma(P_X, \lambda, \alpha, L)$ \quad then \quad $\varepsilon ~>~ 0$.
\end{enumerate}
\label{theo.SanovTRc}
\end{theorem}
\begin{proof}
The theorem could be proven going along the same lines of the proof of Theorem 4 in \cite{BTtit}. We instead provide a proof based on the extension of Sanov's theorem provided in the Appendix (see Theorem \ref{theo.extended_Sanov}). In fact, Theorem \ref{theo.SanovTRc}, as well as Theorem 4 in \cite{BTtit}, can be seen as an application of such a generalized version of Sanov's theorem.

Let us consider
\begin{align}
P_{fn} ~=  \hspace{-0.2cm} & \sum_{(P_{y^n}, P_{\tau^{m_1}}) \in \Gamma^n(\lambda, \alpha, L)} \hspace{-0.7cm} P_X(T(P_{\tau^{m_1}})) P_Y(T(P_{y^n})) \\
= & \sum_{R \in \PP^{m_1}} P_X(T(R)) \hspace{-0.3cm} \sum_{P  \in \Gamma^n(R,\lambda, \alpha,L)}  \hspace{-0.7cm} P_Y(T(P)) \nonumber \\
= & \sum_{R \in \PP^{m_1}} P_X(T(R)) P_Y(\Gamma^n(R,\lambda, \alpha,L)). \nonumber
\end{align}
We start by deriving an upper-bound of the false negative error probability. We can write:
\begin{eqnarray}
   P_{fn} & \leq  & \sum_{R \in \PP^{m_1}} P_X(T(R)) \sum_{P \in \Gamma^n
   (R, \lambda, \alpha, L)} 2^{- n \DD(P || P_Y)} \nonumber \\
   & \leq &  \sum_{R \in \PP^{m_1}} P_X (T(R)) (n +  1)^{|\mathcal X|} 2^{- n
   \min\limits_{P \in \Gamma^n (R, \lambda,\alpha, L)} \hspace{-0.2cm} \DD(P || P_Y)}
   \nonumber \\
   & \le & \sum_{R \in \PP^{m_1}} P_X (T(R)) (n + 1)^{|\mathcal X|} 2^{-
   n \min\limits_{P \in \Gamma(R, \lambda, \alpha, L)} \DD(P || P_Y)}\nonumber\\
   & \leq & (n + 1)^{|\mathcal X|} (m_1 + 1)^{|\mathcal X|} \nonumber \\
   & & \cdot 2^{- n \min\limits_{
   R \in \PP^{m_1} } [\frac{m_1}{n} \DD(R || P_X) + \min\limits_{P \in \Gamma(R, \lambda, \alpha, L)} \DD( P ||  P_Y)]}\nonumber\\
   & \leq & (n + 1)^{|\mathcal X|} (m_1 + 1)^{|\mathcal X|} \nonumber \\
   & & \cdot 2^{- n \min\limits_{
   R \in \PP} [(1-\alpha)c\DD(R || P_X) + \min\limits_{P \in \Gamma(R, \lambda, \alpha, L)} \DD( P ||
   P_Y)]},
   \label{eq.low_bound_P_fn1}
\end{eqnarray}
where the use of the minimum instead of the infimum is justified by the fact that $\Gamma^{n}(R, \lambda, \alpha, L)$ and $\Gamma(R, \lambda, \alpha, L)$ are compact sets. By taking the log and dividing by $n$ we find:
\begin{align}
%- \frac{\log P_{FN}}{n}  \ge \min\limits_{R \in \CC} \big[ & c \DD( R || P_X) + \nonumber\\ & \min\limits_{P
%\in \Gamma_{tr,b} (R, \lambda, L)} \DD( P || P_Y)\big] + \alpha_n,
& - \frac{\log P_{fn}}{n}  ~\ge \nonumber\\
& \hspace{0.5cm} \min\limits_{R \in \PP} \big[ (1-\alpha)c \DD( R || P_X) +  \min\limits_{P
\in \Gamma(R, \lambda, \alpha, L)} \DD( P || P_Y)\big] - \beta_n,
 \label{eq.low_bound_P_fn2}
\end{align}
where $\beta_n = |\XX| \frac{\log(n+1)((1 - \alpha)nc +1)}{n}$ tends to 0 when $n$ tends to infinity.

We now turn to the analysis of a lower bound for $P_{fn}$. Let  $R^*$ be the pmf achieving the minimum in the outer minimisation of eq. (\ref{eq.fnerr_exp_L}). Due to the density of rational numbers within real numbers, we can find a sequence of pmfs' $R_{m_1} \in \PP^{m_1}$ ($m_1 = (1-\alpha) nc$) that tends to $R^*$ when $n$ (and hence $m_1$) tends to infinity. We can write:
\begin{align}
\label{eq.up_bound_P_fn}
P_{fn} & =  \sum_{R \in \PP^{m_1}} P_{X}(T(R))
     P_{Y} (\Gamma^n(R, \lambda, \alpha,  L))\nonumber\\
     & \ge  ~P_{X}(T(R_{m_1}))
     P_{Y} (\Gamma^n(R_{m_1}, \lambda, \alpha, L)),\nonumber\\
     & \ge  ~\frac{2^{- m_1 \DD(R_{m_1} || P_X)}}{(m_1+1)^{|\XX|}}  P_{Y} (\Gamma^n(R_{m_1}, \lambda, \alpha, L)),
\end{align}
%
%where inequalities (a) and (b) derive from a known lower bounds on the probability of a type class \cite{CandT}, and in (c)
where in the first inequality we have replaced the sum with the single element of the subsequence $R_{m_1}$ defined previously, and where the second inequality derives from the well known lower bound on the probability of a type class \cite{CandT}.
From \eqref{eq.up_bound_P_fn}, by taking the log and dividing by $n$, we obtain:
\begin{align}
    & -\frac{\log P_{fn}}{n} ~\le \nonumber\\
    & \hspace{0.3cm}  (1-\alpha)c \DD(R_{m_1} || P_X) - \frac{1}{n}\log P_Y(\Gamma^n(R_{m_1}, \lambda, \alpha, L)) + \beta_n',
\label{eq.up_bound_P_fn_2}
\end{align}
where $\beta_n' = |\XX| \frac{\log(m_1+1)}{n}$ tends to 0 when $n$ tends to infinity.
In order to compute the probability $P_Y(\Gamma^n(R_{m_1}, \lambda, \alpha, L))$, we resort to Corollary \ref{cor.extended_Sanov} of the the generalised version of Sanov's Theorem given in Appendix \ref{sec.appendix.Sanov}.\\
To apply the corollary, we must show that $\Gamma^{n}(R_{m_1}, \lambda, \alpha, L) \overset{H}{\rightarrow} \Gamma(R^*, \lambda, \alpha, L)$.

First of all, we observe that by exploiting the continuity of the $h_c$ function and  the density of rational numbers into the real ones, it is easy to prove that $\Gamma_0^{n}(R_{m_1}, \lambda, \alpha) \overset{H}{\rightarrow} \Gamma_0(R^*, \lambda, \alpha)$.
%\footnote{\MB{I think I would be able to come out with a rigorous proof, which, however is not easy as we say. Given that the proof would be very technical and long and would offer no really new insight, I would leave the text as is, hoping that reviewers will agree with us that proving this statement is indeed easy.}}.
%
Then the Hausdorff convergence of $\Gamma^{n}(R_{m_1}, \lambda, \alpha, L)$ to $\Gamma(R^*, \lambda, \alpha, L)$ follows from the regularity properties of the set of transportation maps stated in Appendix \ref{sec.appendix.cont_MAP}. To see how, we observe that any transformation $S_{PV} \in \AA(L,P)$ mapping $P$ into $V$ can be applied in inverse order through the transformation $S_{VP}(i,j) = S_{PV}(j,i)$. It is also immediate to see that $S_{VP}$  introduces the same distortion introduced by  $S_{PV}$, that is $S_{VP} \in \AA(L,V)$. Let now $P$ be a point in $\Gamma(R^*, \lambda, \alpha, L)$. By definition we can find a map $S_{PV} \in \AA(L,P)$ such that $V \in \Gamma_0(R^*, \lambda, \alpha)$. Since $\Gamma_0^{n}(R_{m_1}, \lambda, \alpha) \overset{H}{\rightarrow} \Gamma_0(R^*, \lambda, \alpha)$, for large enough $n$, we can find a point $V' \in \Gamma_0^{n}(R_{m_1}, \lambda, \alpha)$ which is arbitrarily close to $V$. Thanks to the second part of Theorem \ref{theo_behavior_S} in Appendix \ref{sec.appendix.cont_MAP}, we know that a map $S_{V'P'} \in \AA^n(L,V')$ exists such that $P'$ is arbitrarily close to $P$ and $ P' \in \PP^n$. By applying the inverse map $S_{P'V'}$ to $P'$, we see that $P' \in \Gamma^n(R_{m_1},\lambda, \alpha,L)$, thus permitting us to conclude that, when $n$ increases, $\delta_{\Gamma(R^*, \lambda, \alpha,L)}(\Gamma^n(R_{m_1}, \lambda, \alpha, L)) \rightarrow 0$. In a similar way, we can prove that $\delta_{\Gamma^n(R_{m_1}, \lambda, \alpha, L)}(\Gamma(R^*, \lambda, \alpha,L)) \rightarrow 0$, hence permitting us to conclude that $\Gamma^{n}(R_{m_1}, \lambda, \alpha, L) \overset{H}{\rightarrow} \Gamma(R^*, \lambda, \alpha, L)$.

\label{cor.extended_Sanov}

We can now apply the generalised version of Sanov Theorem as expressed in Corollary \ref{cor.extended_Sanov} of Appendix \ref{sec.appendix.Sanov} to conclude that:
\begin{equation}
   - \lim_{n \rightarrow \infty} \frac{1}{n}~\log P_Y(\Gamma^n(R_{m_1}, \lambda, \alpha, L)) ~= \hspace{-0.3cm} \underset{P \in \Gamma(R^*, \lambda, \alpha, L)}{\min} \hspace{-0.3cm}\DD(P||P_Y).
\label{eq.up_bound_P_fn_Sanov_ext}
\end{equation}

Going back to equation \eqref{eq.up_bound_P_fn_2}, and by exploiting the continuity of the divergence function, we can say that for large $n$ we have:
\begin{align}
\label{eq.up_bound_P_fn_3}
    -\frac{\log P_{fn}}{n} \le (1 - \alpha)c \DD(R^* || P_X) ~+ \hspace{-0.3cm}\min\limits_{P \in \Gamma(R^*, \lambda, \alpha, L)} \hspace{-0.3cm}\DD(P || P_Y) + \nu_n,
\end{align}
where the sequence $\nu_n$ tends to zero when $n$ tends to infinity. By coupling equations  (\ref{eq.low_bound_P_fn2}) and (\ref{eq.up_bound_P_fn_3}) and by letting $n \rarrow \infty$, we eventually obtain:
\begin{align}
    -\lim_{n \rarrow \infty} & \frac{\log P_{fn}}{n} = \nonumber\\
&     \min_{R} [ (1-\alpha)c \cdot \DD(R || P_X) + \min_{P \in \Gamma(R, \lambda, \alpha,  L)} \DD (P || P_Y)],
\label{eq.fn_err_exp}
\end{align}
thus proving the theorem.

\end{proof}

As an immediate consequence of Theorem \ref{theo.SanovTRc}, the set $\Gamma(P_X, \lambda, \alpha, L)$ defines the {\em indistinguishability region} of the test, that is the set of all the sources for which A induces D to decide in favour of $H_0$ even if $H_1$ holds.

\subsection{Analysis of the \SIa~game}
\label{subsec.nontarget}

We now focus on the \SIa~game. For a given choice of $Q(P_{\tau^{m_1}}) \in \SS_{A,T}$ (and hence
$t^m$),  given a sequence $y^n$, the optimum choice of the second part of the attack derives quite easily from the definition of $\Lambda^{n\times m,*}$, namely
\begin{align}
\label{eq.optimum_SA}
    & S^{n,*}_{YZ}(P_{y^n},P_{t^m}) = \\ \nonumber
    & \arg \min\limits_{S^n_{YZ} \in \A^n(L, P_{y^n})} \left(  \min_{Q \in \PP^{m_2}}  h\left( P_{z^n} , \frac{P_{t^m} - \alpha Q}{1 - \alpha}  \right)  \right).
\end{align}
Now the point is to determine the strategy $Q(P_{\tau^{m_1}})$ which maximises the probability that the attack in \eqref{eq.optimum_SA} succeeds. To this purpose, of course, the attacker must exploit the knowledge of $P_Y$. Since solving such a maximisation problem is not an easy task, we will proceed in a different way. We first introduce a simple (and possibly suboptimum) strategy, then we argue that such a strategy is asymptotically optimum, in that the set of the sources that cannot be distinguished from $X$ with this choice is the same set that we have obtained for the \SIat~ setup, which is known to be more favourable to the attacker.
More specifically, we consider the following two-step attacking strategy. In the first step of the attack, A does not know $y^n$, hence he trusts the law of large numbers and optimises $Q(P_{\tau^{m_1}})$ by using $P_Y$ as a proxy for $P_{y^n}$. To do so, he applies equation \eqref{eq.optimum_SA_double}, by replacing $P_{y^n}$ with $P_Y$. Specifically, by indicating with $Q^{\dagger}$, the resulting strategy for the first step of the attack, we have
\begin{align}
\label{eq.asympt_optimum_SA_1step}
   & Q^{\dagger}(P_{\tau^{m_1}}) ~=~
     \arg \min_{Q\in \PP^{m_2}} \\ %\min\limits_{\substack{Q \in \PP^{m_2} \\ S_{YZ} \in \A(L, P_{Y})}} \\
     & \hspace{1.5cm} \min\limits_{\substack{Q' \in \PP^{m_2} \\ S_{YZ} \in \A(L, P_{Y})}} h_c\left( P_{Z} , P_{\tau^{m_1}} + \frac{\alpha}{1 - \alpha} (Q - Q')  \right).
\end{align}
As a by-product of the above minimisation, the attacker also finds the map $S^{n, \dagger}_{YZ}$ representing the optimum attack when $P_{y^n} = P_Y$. Let us indicate the result of the application of such a map to $P_Y$ by $P^{\dagger}_Z$.

In the second part of the attack, A tries to move $P_{y^n}$ as close as possible to $P^{\dagger}_Z$, that is:
\begin{equation}
\label{eq.asympt_optimum_SA_2step}
S^{n,\dagger}_{YZ}(P_{y^n}, P_{t^m}^{\dagger}) ~=~  \arg\min_{S^n_{YZ} \in \A^n(L, P_{y^n})} d(S_Z^n, P_Z^{\dagger}),
\end{equation}
where $S^{n,\dagger}_{YZ}(P_{y^n}, P_{t^m}^{\dagger})$ depends upon the corrupted training sequence obtained after the application of the first part of the attack, namely $P_{t^m}^{\dagger} = (1-\alpha)P_{\tau^{m_1}} + \alpha Q^{\dagger}(P_{\tau^{m_1}})$, through $P_Z^{\dagger}$.

The asymptotic optimality of the strategy ($Q^{\dagger}(P_{\tau^{m_1}})$, $S^{n,\dagger}_{YZ}(P_{y^n}, P_{t^m}^{\dagger})$) derives from the following theorem

\begin{theorem}[Indistinguishability region of the \SIa~ game]
\label{theo_indist_reg_SIa}
The indistinguishability region of \SIa~ game is equal to that of the \SIat~ game (see eq. \eqref{newGamma}) and is asymptotically achieved by the attacking strategy ($Q^{\dagger}(P_{\tau^{m_1}})$, $S^{n,\dagger}_{YZ}(P_{y^n}, P_{t^m}^{\dagger})$).
\end{theorem}
\begin{proof}[Proof (sketch)]
The theorem derives from the observation that due to the law of large numbers, when $n$ grows, $P_{y^n}$ tends to $P_Y$; hence, for large enough $n$, optimising the first part of the attack by replacing $P_{y^n}$ with $P_Y$ does not introduce a significant performance loss. The rigorous proof goes along similar lines to those used to prove Theorem \ref{theo.SanovTRc} and ultimately relies on the continuity of the $h_c$ function and the regularity properties of the set $\A^n(L, P_{y^n})$. The details of the proof are omitted for sake of brevity.
\end{proof}

Given that asymptotic equivalence of the \SIa~ and the \SIat~ games, in the rest of the paper, we will generally refer to the \SIa~game without specifying if we are considering the targeted or non-targeted case.

\section{Source distinguishability for the \SIa~game}
\label{sec.SM_SIa}

In this section, we study the behaviour of the \SIa~game when we vary the decay rate of the false positive error probability $\lambda$. By letting $\lambda$ tend to zero, in fact, we can derive the best achievable performance of the defender when we require only that $P_{fp}$ tends to zero exponentially fast regardless of the decay rate. Then, we use such a result to derive the conditions under which the reliable distinction between two sources is possible in terms of number of corrupted training samples $\alpha$ and maximum allowed distortion $L$.

\subsection{Ultimate achievable performance of the game}

As we said, the goal of this section is to study the limit of the indistinguishability region when $\lambda \rarrow 0$. This limit, in fact, determines all the pmf's $P_Y$ that can not be distinguished from $P_X$ ensuring that the two types of error probabilities tend to zero exponentially fast (with vanishingly small, yet positive, error exponents).

We start by exploiting optimal transport theory to rewrite the indistinguishability region as:
\begin{equation}
\label{set_Gamma_EMD}
\Gamma(P_X,\lambda,\alpha,L)  ~=~  \{ P: ~\exists V \in \Gamma_0(P_X, \lambda, \alpha)  \text{ s.t. } \text{\em EMD}(P,V) \le L\},
\end{equation}
where {\em EMD} (Earth Mover Distance) is the term used in computer vision to denote the minimum transportation cost \cite{rachev1998mass,RTG00}, that is
\begin{equation}
\label{eq.transport_problem}
     \text{\em EMD}(P,V) ~= \min_{S_{PV} : S_{P} = P, S_{V} = V} ~ \sum_{i,j} S_{PV}(i,j) d(i,j).
\end{equation}

With this definition, the main result of this section is stated by the following theorem.
\begin{theorem}
\label{theorem_EMD_L}
Given two sources $X$ and $Y$, a maximum allowed average per-letter distortion $L$ and a fraction $\alpha$ of training samples provided by the attacker, the maximum achievable false negative error exponent $\varepsilon$ for the \SIa~ game is:
\begin{align}
\lim_{\lambda \rarrow 0} ~\lim_{n \rarrow \infty} &- \frac{1}{n} \log ~P_{fn} ~=~  \nonumber \\
&\min_{R} [(1-\alpha)c \DD(R || P_X) ~+ \hspace{-0.3cm} \min_{P \in \Gamma(R, \alpha, L)} \DD (P || P_Y)],
\label{best_e_e}
\end{align}
where $\Gamma(R, \alpha, L) = \Gamma(R, \lambda=0, \alpha, L)$. Accordingly, the ultimate indistinguishability region is given by:
\begin{align}
\label{Gamma_X}
\Gamma(P_X,\alpha, L)  = & \left\{P : ~\exists V \in \Gamma_0(P_X, \alpha) \text{ s.t. } \text{\em EMD}(P,V) \le L\right\},
\end{align}
where $\Gamma_0(P_X, \alpha) = \Gamma_0(P_X, \lambda = 0, \alpha)$. Moreover, $\Gamma(P_X,\alpha, L)$ can be rewritten as:
\begin{align}
\label{Gamma_X_1}
\Gamma(P_X,\alpha, L)  = & \bigg\{ P : \min_{V: \text{\em EMD}(P,V) \le L} \sum_{i} \left[V(i) \text{--} P_X(i) \right]^+ \le \frac{\alpha}{(1 - \alpha)}\bigg\}\nonumber\\
 = & \bigg\{ P : \min_{V: \text{\em EMD}(P,V) \le L} d_{L_1}(V,P_X) ~\le~ \frac{2\alpha}{(1 - \alpha)}\bigg\}.\nonumber\\
\end{align}
with $[a]^+ =  \max\{a,0\}$.
\end{theorem}
\begin{proof}

The proof of the first part goes along the same steps used in the proof of Theorems 3 and 4 in \cite{BT_SMargin} and is not repeated here. We show, instead, that $\Gamma(P_X,\alpha, L)$ can be rewritten as in \eqref{Gamma_X_1}.

By observing that $h_c(P,Q) = 0$ if and only if $P=Q$, it is immediate to see that the set $\Gamma_0(P_X, \lambda=0,\alpha)$ takes the following expression:
\begin{equation}
\Gamma_0(P_X, \alpha) = \{P : ~\exists Q, Q' \in \PP \text{ s.t. } P   ~=~ P_X + \frac{\alpha }{(1 - \alpha)} (Q - Q')\}.
\label{Gamma_X_var}
\end{equation}
Expression \eqref{Gamma_X_var} can be rewritten by avoiding the introduction of the auxiliary pmf's $Q$ and $Q'$. To do so, we observe that $Q(i)$ must be larger than $Q'(i)$ for all the bins $i$ for which $P(i) > P_X(i)$ (and viceversa). In addition, $Q$ and $Q'$ must be valid pmf's, hence we have $\sum_i [ Q(i) - Q'(i)]^+ = \sum_i [ Q'(i) - Q(i)]^+ \le 1$.
%From the arbitrariness of the choice of each element $C'(i)$ ($C(i)$),
Then, it is easy to see that \eqref{Gamma_X_var} is equivalent to the following definition:
\begin{align}
\label{Gamma_X_2}
\Gamma_0(P_X, \alpha) & =  \left\{P : ~\sum_{i} \left[P(i) - P_X(i) \right]^+  \le \frac{\alpha}{(1 - \alpha)}\right\}\\
& =  \left\{P : ~d_{L_1}(P,P_X) \le \frac{2\alpha}{(1 - \alpha)}\right\},\nonumber
\end{align}
where the second equality follows by observing that $d_{L_1}(P,P_X) = \sum_i [P(i) - P_X(i)]^+ + \sum_i [P_X(i) - P(i)]^+$.
Eventually, equation \eqref{Gamma_X_1} derives immediately from the expression of $\Gamma_0(P_X, \alpha)$ given in \eqref{Gamma_X_2}.
\end{proof}
According to Theorem \ref{theorem_EMD_L}, $\Gamma(P_X,\alpha, L)$ provides the {\em ultimate indistinguishability region} of the test, that is the set of all the pmf's for which A wins the game.

Before going on, we pose to discuss the geometrical meaning of the set $\Gamma_0(P_X, \alpha)$ in \eqref{Gamma_X_var}. To do so, we introduce the set $\Lambda_0^*$, obtained from $\Lambda^*$ by letting $\lambda \rightarrow \infty$:
\begin{equation}
\Lambda_0^* = \left\{ (P, P'): ~\exists Q \text{ s.t. }  P' ~=~ \frac{P - \alpha Q}{(1-\alpha)}\right\}.
\label{lambda_ultimate_noargument}
\end{equation}
As usual, we can fix the pmf $P$ and define:
\begin{equation}
\Lambda_0^*(P) = \left\{ P': ~\exists Q \text{ s.t. }  P' ~=~ \frac{P - \alpha Q}{(1-\alpha)}\right\}.
\label{lambda_ultimate}
\end{equation}
By referring to Figure \ref{fig.GeometricSets} (left part), we can geometrically interpret $\Lambda_0^*(P)$ as the set of the pmf's $P'$ such that $P$ is a convex combination (with coefficient $\alpha$) of $P'$ with a point $Q$ of the probability simplex.
Starting from \eqref{newGamma0}, we can then rewrite  $\Gamma_0(P_X, \alpha)$ as follows:
\begin{equation}
\Gamma_0(P_X, \alpha)  = \{P : ~\exists  Q \in \PP \text{ s.t. } P \in \Lambda_0^*((1 - \alpha) P_X + \alpha Q)\}.
\label{Gamma_X_var_lambda}
\end{equation}
Accordingly, $\Gamma_0(P_X, \alpha)$ is geometrically obtained as the union of the acceptance regions built from the points which can be written as a convex combination of $P_X$ with some point $Q$ in the simplex. As shown in the right part of Figure \ref{fig.GeometricSets}, such a region corresponds to an hexagon centred in $P_X$, which, in the probability simplex, is equivalent to the set of points whose $L_1$ distance from $P_X$ is smaller than or equal to $2\alpha/(1-\alpha)$ (as stated in \eqref{Gamma_X_2}). Of course, only the points of the hexagon that lie inside the simplex are valid pmf's and then must be accounted for.

A pictorial representation of the set $\Gamma(P_X, \alpha,L)$ is given in Figure \ref{fig.GeometricSets2}.

\begin{figure}[t!]
\centering \includegraphics[width = 0.49\columnwidth]{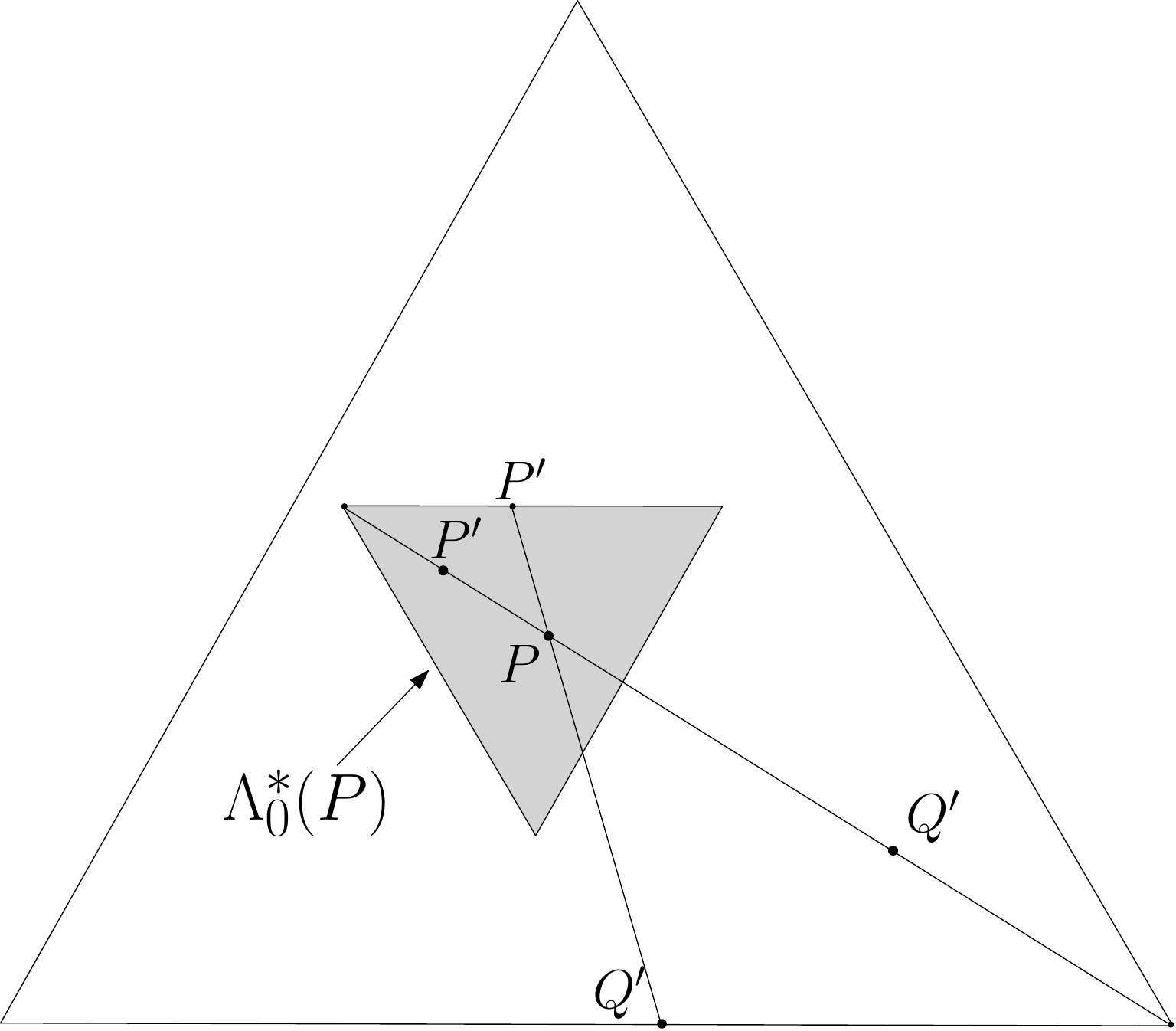} \includegraphics[width = 0.49\columnwidth
]{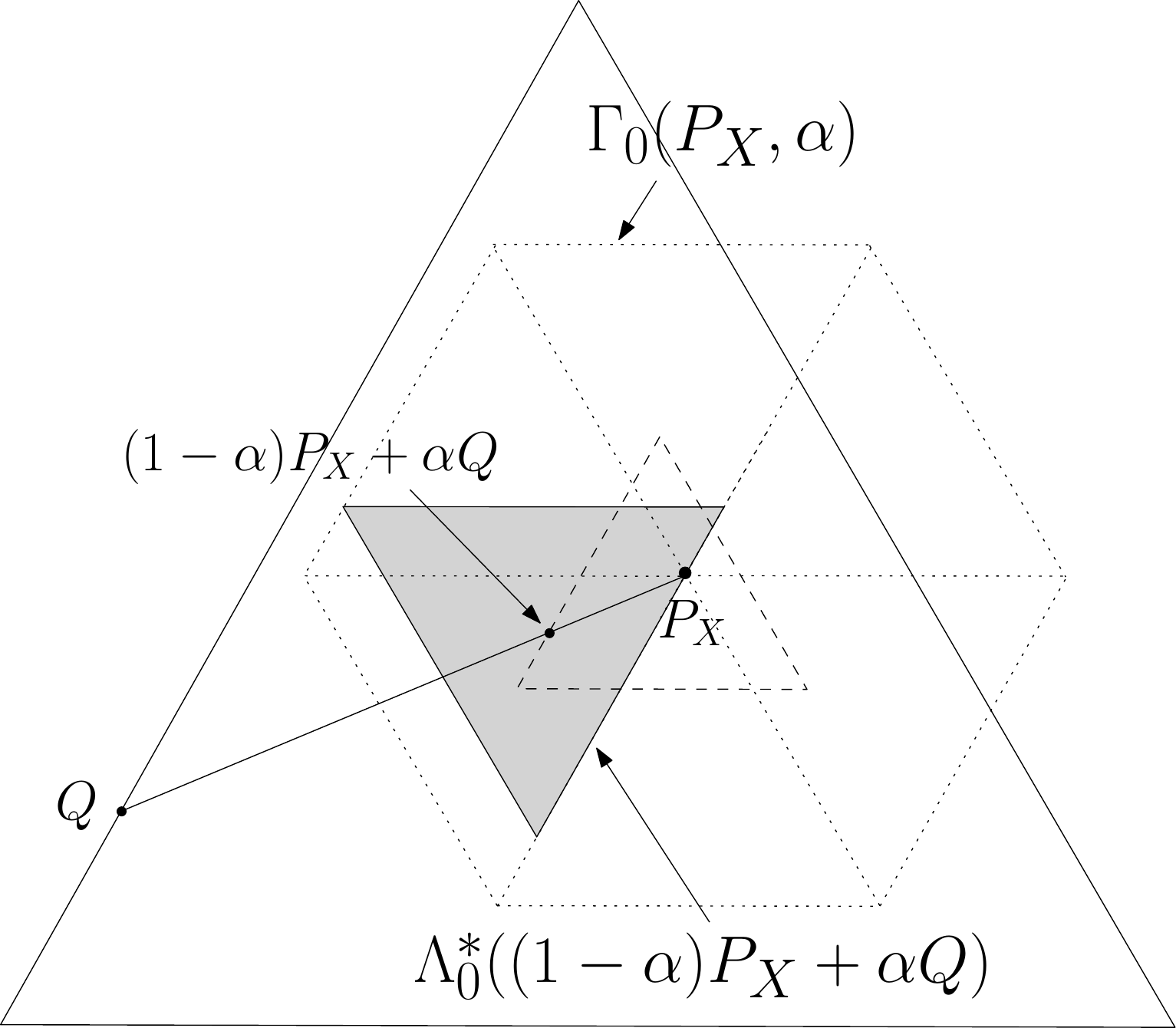}
\caption{Geometrical interpretation of $\Lambda^*_0(P)$ (left) and geometrical construction of $\Gamma_0(P_X, \alpha)$ (right). The size of the sets are exaggerated for graphical purposes.}
\label{fig.GeometricSets}
\end{figure}

\begin{figure}[t!]
\centering \includegraphics[width = 0.58\columnwidth]{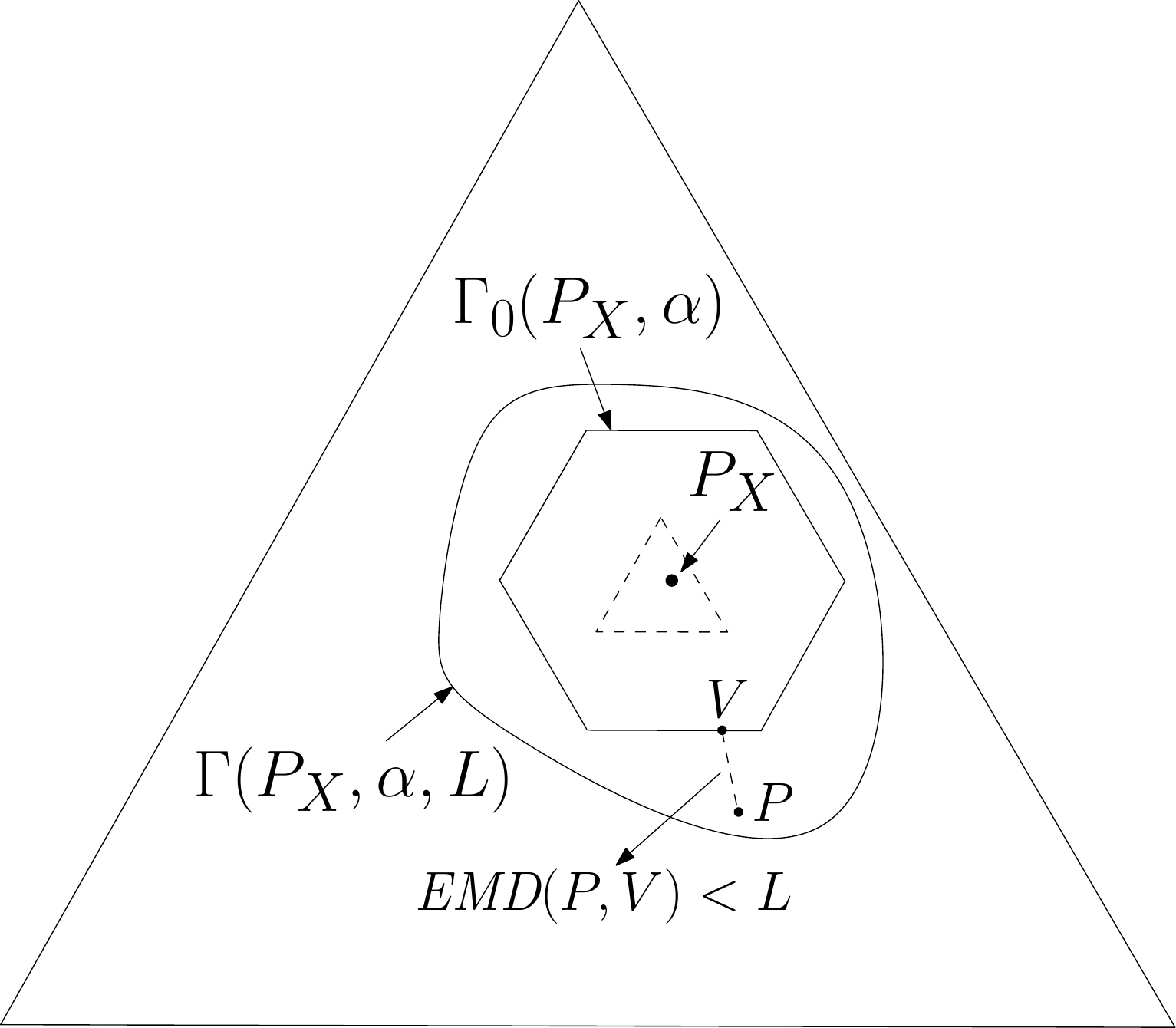}
\caption{Geometrical interpretation of $\Gamma(P_X, \alpha, L)$ as stated in Theorem \ref{theorem_EMD_L}.}
\label{fig.GeometricSets2}
\end{figure}

\subsection{Security margin and blinding corruption level ($\alpha_b$)}
\label{sec.SM_alpha}

By a closer inspection of the {\em ultimate indistinguishability region} $\Gamma(P_X,\alpha, L)$, we can derive some interesting parameters characterising the distinguishability of two sources in adversarial setting.
Let $X \sim P_X$ and $Y \sim P_Y$ be two sources.
Let us focus first on the case in which the attacker can not modify the test sequence ($L = 0$). In this situation, the ultimate indistinguishability region boils down to $\Gamma_0(P_X,\alpha)$.
%By rewriting the sum in \eqref{Gamma_X_2} as a function of the $L_1$ distance between $P$ and $P_X$,
Then we conclude that D can tell the two sources apart if $d_{L_1}(P_Y, P_X) > \frac{2 \alpha}{(1 - \alpha)}$. On the contrary, if $d_{L_1}(P_Y, P_X) \le \frac{2 \alpha}{(1 - \alpha)}$, A is able to make the sources indistinguishable by corrupting the training sequence. Clearly, the larger the $\alpha$ the easier is for A to win the game. We can define the {\em blinding corruption level} $\alpha_{b}$, as the minimum value of $\alpha$ for which two sources $X$ and  $Y$ can not be distinguished. Specifically, we have:
\begin{align}
\alpha_b(P_X, P_Y) ~=~  \frac{d_{L_1}(P_Y, P_X)}{2 + d_{L_1}(P_Y, P_X)} ~=~ \frac{\sum_i \left[P_Y(i) - P_X(i)\right]^+}{1 + \sum_i \left[P_Y(i) - P_X(i)\right]^+}.
\label{alpha_blind_D_zero}
\end{align}
From \eqref{alpha_blind_D_zero} it is easy to see that $\alpha_b$ is always lower than $1/2$, with the limit case $\alpha_b = 1/2$ corresponding to a situation in which $P_X$ and $P_Y$ have completely disjoint supports\footnote{We remind that for any pair of pmf's $(P,Q)$, $d_{L_1}(P,Q) ~\le~ 2$.}.
%\BT{It is interesting to notice that $\alpha_b$ is symmetric with respect to the two sources. Since the attacker is allowed only to add samples to the training sequence (and not remove existing ones), this might seem counterintuitive. Actually, the symmetry of the blinding value is a consequence of the worst case approach adopted by the defender; according to our formulation, in fact, the defender makes a worst case assumption on the corruption strategy in order to ensure/meet the false positive constraint (see definition \eqref{eq.SD}), thus making the value of $\alpha_b$ symmetric.}
It is interesting to notice that $\alpha_b$ is symmetric with respect to the two sources. Since the attacker is allowed only to add samples to the training sequence without removing existing samples, this might seem a counterintuitive result. Actually, the symmetry of $\alpha_b$ is a consequence of the worst case approach adopted by the defender. In fact, D itself discards a subset of samples from the training sequence in such a way to maximise the probability that the remaining part of the training sequence and the test sequence have been drawn from the same source.

%\enlargethispage{\baselineskip}

Let us now consider the more general case in which $L\neq 0$.  For a given $\alpha < \alpha_b$, we look for the maximum distortion allowed to $A$ for which it is possible to reliably distinguish between the two sources.
%By combining equation \eqref{Gamma_X} and \eqref{Gamma_X_2},
From equation \eqref{Gamma_X_1}, we see that the attack does not succeed if:
\begin{equation}
\min_{V: \text{\em EMD}(P_Y,V) \le L} d_{L_1}(V,P_X) ~>~ \frac{2 \alpha}{(1 - \alpha)}.
\label{eq.new}
\end{equation}
This leads to the following definition, which extends the concept of security margin, introduced in \cite{BT_SMargin}, to the more general setup considered in this paper.
\begin{definition}[Security Margin in the \SIa~ setup]
Let $X \sim P_X$ and $Y \sim P_Y$ be two discrete memoryless sources. The maximum distortion allowed to the attacker for which the two sources can be reliably distinguished in the \SIa~ setup with a fraction $\alpha$ of possibly corrupted samples, is called Security Margin and is given by
\begin{align}
\SS \MM_{\alpha}& (P_X, P_Y) ~=~ L_{\alpha}^*,
\label{security_margin}
\end{align}
where $L_{\alpha}^* = 0$ if $P_Y \in \Gamma_0(P_X, \alpha)$, while, if $P_Y \notin \Gamma_0(P_X, \alpha)$, $L_{\alpha}^*$ is the quantity which satisfies
\begin{align}
\min_{V : \text{\em EMD}(P_Y,V) \le L_{\alpha}^*} d_{L_1}(V, P_X)  ~=~ \frac{2\alpha}{(1 - \alpha)}.
\label{equationSM}
\end{align}
\label{def.SM}
\end{definition}
\begin{figure}[t!]
\centering \includegraphics[width = 0.55\columnwidth]{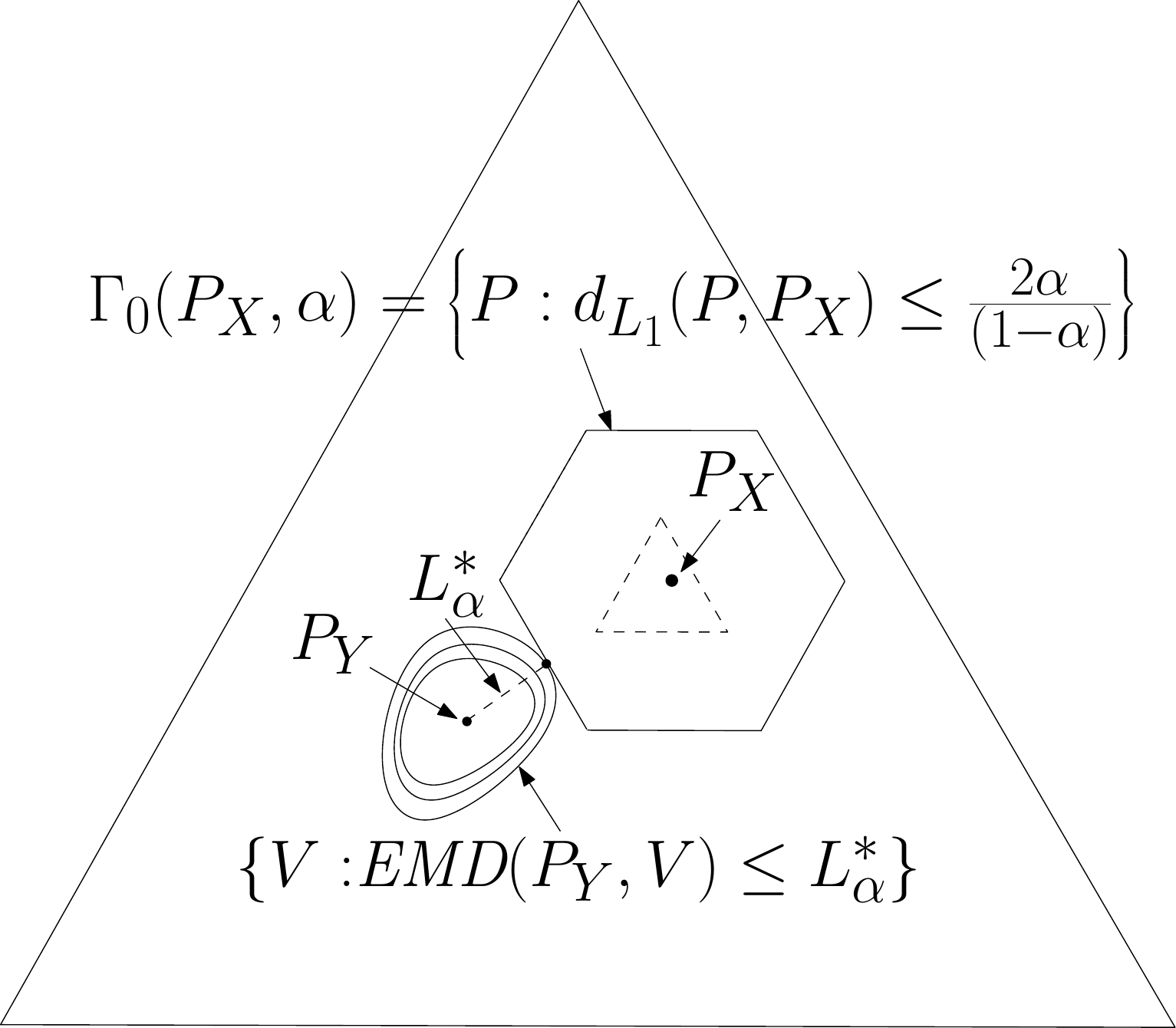}
\caption{Geometrical interpretation of the Security Margin between two sources $X$ and $Y$.}
\label{fig.SM_Geometric_interpretation}
\end{figure}
\noindent A geometric interpretation of $L^*_{\alpha}$ is given in Figure \ref{fig.SM_Geometric_interpretation}. By focusing on the case $P_Y \notin \Gamma_0(P_X, \alpha)$, and by observing that
\begin{equation}
    \min_{V : \text{\em EMD}(P_Y,V) \le L} d_{L_1}(V, P_X)
\label{eq.explicit_min}
\end{equation}
is a monotonic non-increasing function of $L$, the security margin  can be expressed in explicit form as
\begin{align}
\SS \MM_{\alpha} (P_X, P_Y) = \underset{L'}{\arg\min} \min_{V: \text{\em EMD}(P_Y,V) \le L'} \left| d_{L_1}(V,P_X) - \frac{2\alpha}{(1 - \alpha)} \right|.
\label{security_margin_explicit}
\end{align}
When $L > \SS \MM_{\alpha}(P_X, P_Y)$, it is not possible for D to distinguish between the two sources with positive error exponents of the two kinds.

By looking at the behavior of  the security margin as a function of $\alpha$, we see that $\SS \MM_{\alpha_b}(P_X, P_Y) = 0$, meaning that, whenever the fraction of corrupted samples reaches the critical value, the sources can not be distinguished even if the attacker does not introduce any distortion. On the contrary, setting $\alpha = 0$ corresponds to study the distinguishability of the sources with uncorrupted training; in this case we have $\SS \MM_{0}(P_X,P_Y) = \text{\em EMD}(P_X,P_Y)$, in agreement with \cite{BT_SMargin}. With reference to Figure \ref{fig.SM_Geometric_interpretation}, it is easy to see that when $\alpha = 0$ the hexagon representing $\Gamma_0(P_X, \alpha)$ collapses into the single point $P_X$ and the security margin corresponds to the Earth Mover Distance between $Y$ and $X$. Eventually, we notice that, for $\alpha > 0$, the value of the security margin in \eqref{security_margin_explicit} is less than $\text{\em EMD}(P_X,P_Y)$. This is also an expected behaviour since the general setting considered in this paper is more favourable to the attacker than the setting in \cite{BT_SMargin}.

By looking at \eqref{security_margin_explicit}, we can argue that
the Security Margin is symmetric with respect to the two sources $X$ and $Y$, that is, $\SS \MM_{\alpha}(P_Y,P_X) = \SS \MM_{\alpha}(P_X,P_Y)$.

To show that this is the case, we observe that the pmf $V'$ associated with the minimum $L$, for which we have $\text{\em EMD}(P_Y, V')= \SS \MM_{\alpha}(P_X,P_Y)$, can be obtained through the application of a map $S_{P_Y V}$ that works as follows: it does not modify a portion $\alpha/(1-\alpha)$ of $P_Y$ and moves the remaining mass into an equal amount of $P_X$ in a convenient way (i.e., in such a way to minimise the overall distance between the masses). The inverse map can be applied to bring the same quantity of mass from $P_X$ to $P_Y$, while leaving as is the remaining mass, thus obtaining a $V''$ which satisfies $\text{\em EMD}(P_X, V'')= \text{\em EMD}(P_Y, V')$ (because of the symmetry of the per-symbol distortion $d$) and $d_{L_1}(V'', P_Y)  = d_{L_1}(V', P_X) = 2\alpha/(1-\alpha)$. Arguably, $V''$ is the pmf for which $\text{\em EMD}(P_X, V'')= \SS \MM_{\alpha}(P_Y,P_X)$; hence, $\SS \MM_{\alpha}(P_Y,P_X) = \SS \MM_{\alpha}(P_X,P_Y)$.

\subsubsection{Bernoulli sources}
\label{sec_bernoulli}

In order to get some insights on the practical meaning of $\alpha_b$ and $\SS\MM_{\alpha}$, we consider the simple case of two Bernoulli sources with parameter $q = P_X(1)$ and $p = P_Y(1)$.
%Let us suppose that $p > q$ (w.l.o.g.).
Assuming that no distortion  is  allowed to the attacker,
the minimum fraction of samples that A must add to induce a decision error is, according to (\ref{alpha_blind_D_zero}), $\alpha_b= \frac{|p - q|}{1 + |p - q|}$. For instance, and rather obviously, when $|p - q| = 1$, to win the game A must introduce a number of fake samples equal to the number  of samples of the correct training sequence, i.e. $\alpha = 0.5$. With regard to $\SS\MM$, we have:
%the security margin is obtained by solving
%%
%$$\min_{r: |r - p| = \SS\MM_{\alpha}(p,q)} \left(|r - q| - \frac{\alpha}{1 - \alpha}\right)^+,$$
%
%from which we get:
%
\begin{equation}
\label{SM_bernoulli_1}
\SS\MM_{\alpha}(p,q) =  \left\{
\begin{array}{ll}
 |q - p| - \frac{\alpha}{1 - \alpha} & \alpha ~<~ \alpha_b \\
0 & \alpha ~\ge~ \alpha_b
\end{array}
\right..
\end{equation}
%
%The geometrical meaning of \eqref{SM_bernoulli_1} is illustrated in Figure \ref{fig.SM_geometric_Bernoulli} for two generic Bernoulli sources with $p > q$. If $\alpha = 0$, we get the same expression of the security margin for the uncorrupted training case derived in \cite{BT_SMargin}.
Figure \ref{fig.SMcurve_Bernoulli} illustrates the behavior of $\SS\MM_{\alpha}(p,q)$ as a function of $\alpha$ when $p = 0.3$ and $q = 0.7$. The blinding corruption value is $\alpha_b = 0.286$.

%\begin{figure}[t!]
%\centering \includegraphics[width = 0.7\columnwidth]{SM_interpretation.pdf}
%\caption{Geometrical interpretation of the Security Margin between $X$ and $Y$. When $\alpha = 0$, $\Gamma_0(q,\alpha)$ boils down to point $p$ and $\SS\MM= (q - p)$ (see \cite{BT_SMargin}).}
%\label{fig.SM_geometric_Bernoulli}
%\end{figure}

\begin{figure}[t!]
\centering \includegraphics[width = 0.55\columnwidth]{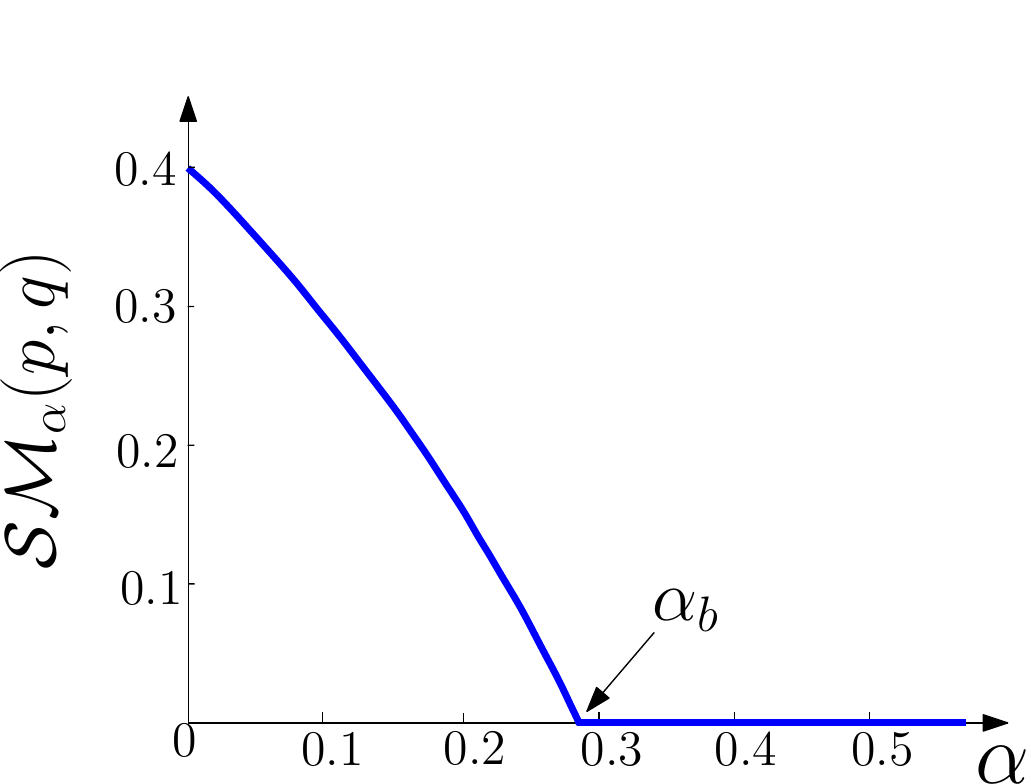}
\caption{Security margin as a function of $\alpha$ for Bernoulli sources with parameters $p = 0.3$ and $q=0.7$ ($\alpha_b = 0.286$).}
\label{fig.SMcurve_Bernoulli}
\end{figure}

\section{Source identification game with replacement of training samples}
\label{sec.SI_CTR_c}

In this section, we study a variant of the game with corrupted training, in which A observes the training sequence and can replace a selected fraction of samples.
Let $\tau^m$ indicate the original $m$-sample long training sequence drawn from $X$ and let $\MM$ be a subset of $m_2~=~\alpha m$ indexes in $[1, 2 \dots m]$. The attacker can choose the index set $\MM$ and replace the corresponding samples with $m_2$ fake samples. More formally, given the original training sequence $\tau^m$, the training sequence {\em seen} by the defender is $t^m = \sigma(\tau_{\bar{\sMM}}^{m_1}|| \tau^{m_2})$, where $\bar{\MM}$ is the complement of $\MM$ in $[1, 2 \dots m]$, $\tau_{\bar{\sMM}}^{m_1}$ is the set of original (non-attacked) samples, and $\tau^{m_2}$ is the sequence with the fake samples introduced by the attacker.

Figure \ref{fig.ADVsetup2} illustrates the adversarial setup considered in this section for the case of a targeted attack. Arguably, this scenario is more favourable to the attacker with respect to the \SIa~game.

\begin{figure}[t!]
\centering \includegraphics[width = 0.99\columnwidth]{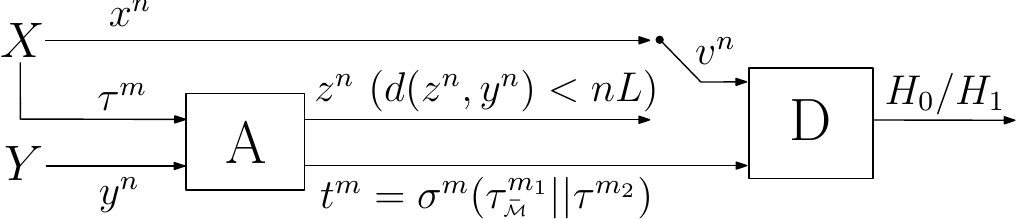}
\caption{Block diagram of the \SIr~ game (targeted corruption). Given the original training sequence $\tau^m$, the adversary has the possibility to replace a selected subset of $m_2$ training samples with fake ones.}
\label{fig.ADVsetup2}
\end{figure}

\subsection{Formal definition of the \SIr~game}

In the sequel, we formally define the source identification game with replacement of selected samples, namely the \SIr~game. As anticipated, we focus on a version of the game in which the corruption of the training samples depends on the to-be-attacked sequence $y^n$ (targeted attack), the extension to the case of non-target attack, in fact, can be easily obtained by following the same approach used in Section \ref{subsec.nontarget}.

\subsubsection{Defender's strategies}

As in the \SIa~game, in order to be sure that the false positive error probability is lower than $2^{-n\lambda}$, the defender adopts a worst case strategy and considers the maximum of the false positive error probability over all the possible $P_X$ and over all the possible attacks that the training sequence may have undergone, yielding:
\begin{equation}
\label{eq.SD_r}
 \SS_{D} = \{ \Lambda^{n \times m} \subset \PP^n \times \PP^m: ~\max_{P_X \in \mathcal{P}} ~\max_{s \in \SS_{A,T}} P_{fp} ~\le~ 2^{-\lambda n}\}.
\end{equation}
While the above expression is formally equal to that of the \SIa~game (see eq. \eqref{eq.SD_bis}), the maximisation over $\SS_{A,T}$ is now more cumbersome, due to the additional degree of freedom available to the attacker, who can selectively remove the samples of  the original training sequence.  In fact, even if D knew the position of the corrupted samples, simply throwing them away would not guarantee that the remaining part of the sequence would follow the same statistics of $X$, since the attacker might have deliberately altered them by selectively choosing the samples to replace.

\subsubsection{Attacker's strategies}

With regard to the attacker, the part of the attack working on the test sequence $y^n$ is the same as for the \SIa ~case, while the part regarding the corruption of the training sequence must be redefined.
To this purpose, we observe that the corrupted training sequence may be any sequence $t^{m}$ for which $d_H(t^m, \tau^m) \le \alpha m$, where $d_H$ denotes the Hamming distance.
Given that the defender basis his decision on the type of $t^m$, it is convenient to rewrite the constraint on the Hamming distance between sequences as a constraint on the $L_1$ distance between the corresponding types.
In fact, by looking at the empirical distribution of the corrupted sequence, searching for a sequence $t^m$ s.t.  $d_{H}(t^m, \tau^m) \le \alpha m$ is equivalent to searching for a pmf $P_{t^m} \in \PP^m$ for which  $d_{L_1}(P_{t^m}, P_{\tau^m}) \le 2\alpha$  (see the proof of Lemma 2 in  \cite{BT13}).
Therefore, the set of strategies of the attacker is defined by $ \SS_{A} = \SS_{A,T} \times  \SS_{A,O}$, where
\begin{align}
    \SS_{A,T} ~=~ \{ & Q(P_{\tau^{m}}, P_{y^n}): ~\PP^m \times \PP^n \rightarrow \PP^m  \nonumber \\
    & \text{such that} ~d_{L_1}(Q(P_{\tau^{m}}, P_{y^n}), P_{\tau^m}) ~\le~ 2\alpha \}, \label{eq.SA_Tr_r_1} \\
    \SS_{A,O} ~=~ \{ & S^n_{YZ}(P_{y^n}, P_{t^m}): ~\PP^n \times \PP^m \rightarrow \A^n(L, P_{y^n}) \}.
\label{eq.SAD_TR_r}
\end{align}
Note that, in this case, the function $Q(\cdot, \cdot)$ gives the type of the whole training sequence observed by D (not only the fake subpart, as it was in the \SIa~ game),  that is, $P_{t^m} = Q(P_{\tau^m}, P_{y^n})$.

In the following, we will find convenient to express the attacking strategies in $\SS_{A,T}$ in an alternative way.
Since the attacker {\em replaces} the samples of a subpart of the training sequence, the corruption strategy is equivalent to first removing a subpart of the training sequence and then adding
a fake subsequence of the same length. Then, the sequence is reordered to hide the position of the fake samples. By focusing on the type of the observed training sequence, we can write:
\begin{equation}
P_{t^m} ~=~ P_{\tau^m} - \alpha Q_R(P_{\tau^m}, P_{y^n}) + \alpha Q_A(P_{\tau^m}, P_{y^n}).
\label{eq.alternative}
\end{equation}
where $Q_R(P_{\tau^m}, P_{y^n})$ and $Q_A(P_{\tau^m}, P_{y^n})$ (both belonging to $\PP^{m_2}$) are the types of the removed and injected subsequences respectively. In order to simplify the notation, in the following we will avoid to indicate explicitly the dependence of $Q_R(P_{\tau^m}, P_{y^n})$ and $Q_A(P_{\tau^m}, P_{y^n})$ on $P_{\tau^m}$, $P_{y^n}$,
%and will indicate them as $Q_R()$ and $Q_A()$ or even $Q_R$ and $Q_A$.
and will indicate them as $Q_R()$ and $Q_A()$. Furthermore, we will use notation $Q_R$ and $Q_A$ whenever the dependence from the arguments is not relevant.
By varying $Q_R$ and $Q_A$, we obtain all the pmf's that can be produced from $P_{{\tau}^m}$ by first removing and later adding $m_2$ samples. Of course not all pairs $(Q_R,~Q_A)$ are admissible since the $P_{t^m}$ resulting from eq. \eqref{eq.alternative} must be a valid pmf, i.e. it must be nonnegative for all the symbols of the alphabet $\XX$.

%$\SS_{A,T}$ can then be rewritten as:
%%
%\begin{align}
%    \SS_{A,T}' = & \left\{ (Q_R(P_{\tau^{m}}, P_{y^n}), Q_A(P_{\tau^{m}}, P_{y^n})) ~\text{ s. t. } \right. \nonumber\\
%    & \hspace{2cm} \left. P_{\tau^m} - \alpha (Q_R - Q_A) \in \PP^m \right\}.
%  \label{eq.SA_Tr_r_2}
%\end{align}
%%
%that is, the set of all the pmf's that can be obtained from $P_{{\tau}^m}$ by first removing and later adding $m_2$ samples.

%{\em It must be noticed that, not all the pairs $(Q_R,Q_A)$ which result in a valid pmf $P_{t^m}$ ($P_{t^m} = P_{\tau^m} - \alpha (Q_R - Q_A)$) are valid strategies for the removal and addition. In fact, given a difference $(Q_R - Q_A)$ producing a valid pmf in $\PP^m$, to find the admissible pairs we have to impose that after the removal $P_{\tau^m} - \alpha Q_R > 0$ for all the alphabet symbols} \MB{I think this comment is not true, if the final pmf is a valid one then there is no problem. Doublecheck it.}.

%Choosing a pmf $Q(P_{\tau^m}, P_{y^n})$ in $\SS_{A,T}$ is indeed equivalent to choose a pair $(Q_R(P_{\tau^m}, P_{y^n}),Q_A(P_{\tau^m}, P_{y^n}))$ in $\SS_{A,T}'$ and then consider the pmf $P_{t^m} = P_{\tau^m} - \alpha (Q_R - Q_A)$ (see also the proof of Theorem \ref{theorem_EMD_L} in Section \ref{sec.SM_alpha}).
%
%Whereas in general definition $\SS_{A,T}$ is preferable, due to its more compact form, and it is adopted for describing the attack, from the defender'e perspective it is more convenient to refer to definition $\SS_{A,T}'$.
%%is more convenient to think about the corruption as a sequence of removal and addition of samples (definition $\SS_{A,T}'$).

\subsubsection{Payoff}

As usual, the payoff function is defined as
\begin{equation}
%    u(\Lambda^{n \times m}_c, (Q(P_{\tau^{m}}), S^n_{YZ}(P_{y^n}, P_{t^m}))) = -P_{fn}.
u(\Lambda^{n \times m}, (Q_R(), Q_A(), S^n_{YZ}())) ~=~ -P_{fn}.
\label{eq.payoff_TR_r}
\end{equation}

\subsection{Equilibrium point and payoff at the equilibrium}

In order to ensure that  $P_{fp}$ is always lower than $2^{- \lambda n}$ , it is convenient to use the attack formulation given in \eqref{eq.alternative}. For a given $P_X$, $Q_R$ and $Q_A$, $P_{fp}$ is the probability that $X$ generates two sequences $x^n$ and $\tau^m$, such that the pair of type classes $(P_{x^n}, P_{\tau^m} - \alpha(Q_R() - Q_A()))$ falls outside $\Lambda^{n \times m}$. Accordingly, the set of strategies available to D can be rewritten as:
\begin{align}
\label{eq.SD_explicit_r}
&\SS_D  = \bigg \{ \Lambda^{n \times m} : \max_{P_X \in \PP} ~\max_{Q_R(),Q_A()} \sum_{P_{y^n} \in \PP^n} P_Y(T(P_{y^n})) \cdot   \\ \nonumber
&\sum_{(P_{x^n}, P_{t^m}) \in \bar{\Lambda}^{n \times m}} \hspace{-0.6cm} P_X(T(P_{x^n})) ~\cdot \hspace{-1.1cm}   \sum_{\substack{P_{\tau^{m}} \in \PP^{m}:  \\ P_{\tau^m} - \alpha(Q_R() - Q_A()) = P_{t^m}}} \hspace{-0.9cm}P_X(T(P_{\tau^{m}})) ~\le~ 2^{-\lambda n}
\bigg \}.
\end{align}

By proceeding as in the proof of Lemma \ref{lemma.optimum_SD}, it is easy to prove that the asymptotically optimum strategy for the defender corresponds to the following:
\begin{align}
\label{eq.optimum_SD_r}
    \Lambda^{n \times m,*}  & ~=~ \big\{ (P_{x^n}, P_{t^m})  :  \nonumber \\
    & \min_{Q_R,Q_A \in \PP^{m_2}} h\left(P_{x^n},P_{t^m} + \alpha (Q_R - Q_A) \right) \le~ \lambda -  \delta_{n} \big\},
\end{align}
where $\delta_{n}$ tends to 0 as $n \rightarrow \infty$ and the minimization is limited to $Q_R$ and $Q_A$ in $\PP^{m_2}$ such that $P_{t^m}+\alpha (Q_R - Q_A)$ is a valid pmf. Consequently, the optimum attacking strategy is given by:
\begin{align}
    & (Q^*(P_{\tau^{m}}, P_{y^n}),  ~S^{n,*}_{YZ}(P_{y^n},P_{t^m})) ~= \nonumber \\
   &  \operatornamewithlimits{argmin} \limits_{\substack{P_{t^m} \text{ s.t. } d_{L_1}(P_{t^m}, P_{\tau^m}) \le 2\alpha \\ S^n_{YZ} \in \A^n(L, P_{y^n})}}
    \bigg[ \min_{Q_R,Q_A}  h\left( P_{z^n} , P_{t^m} + \alpha (Q_R - Q_A) \right) \bigg],
    \label{eq.optimum_SA_r}
\end{align}
hence resulting in the following theorem.
\begin{theorem}
\label{teo.SIR}
The \SIr~ game with targeted corruption is a dominance solvable game, whose only rationalizable equilibrium corresponds to the profile $(\Lambda^{n \times m,*}, (Q^*(), ~S^{n,*}_{YZ}()))$ given by equations \eqref{eq.optimum_SD_r} and \eqref{eq.optimum_SA_r}.
\end{theorem}

%\MB{It is not necessary to explicitly write it}
%{\em For the case $L=0$, we get the optimum strategy of corruption of the training, which is
%%
%%
%\begin{align}
%     \label{eq.optimum_SA_Tr_r}
%    Q^*(P_{\tau^{m}}, P_{y^n}) = &
%    \arg\min_{P_{t^m} \text{ s.t. } d_{L_1}(P_{t^m}, P_{\tau^m}) \le 2\alpha} \\ \nonumber
%    & \hspace{1.5cm} \min_{Q_R,Q_A}  h\left( P_{y^n} , P_{t^m} + \alpha (Q_R - Q_A) \right).
%\end{align}}
%

In order to study the asymptotic payoff of the \SIr~ game at the equilibrium, we parallel the analysis carried out in Sec. \ref{subsec.SIapayoff}. By considering the case $L=0$, the set of pairs of types for which D will accept $H_0$ as a consequence of the attack to the training sequence is given by
\begin{align}
\label{gamma_n_c_c}
\Gamma_0^n(\lambda,\alpha) ~=~ \{ & (P_{y^n}, P_{\tau^m}) : \nonumber \\
 &\exists P_{t^m} \text{ s.t. } d_{L_1}(P_{t^m}, P_{\tau^m}) ~\le~ 2\alpha \nonumber \\
 & \text{and } (P_{y^n}, P_{t^m}) ~\in~ \Lambda^{n \times m,*}\}.
\end{align}
If we fix the type of the original training sequence, we get:
\begin{align}
\label{gamma_n_c_c_fix}
\Gamma_0^n(P_{\tau^m}, \lambda,\alpha)~ &=~ \{P_{y^n}: ~\exists P_{t^m} \text{ s.t. } d_{L_1}(P_{t^m}, P_{\tau^m}) ~\le~ 2\alpha \nonumber \\
& \hspace{1.4cm} \text{and } P_{y^n} \in \Lambda^{n,*}(P_{t^m}) \} \nonumber \\
& = ~ \{ P_{y^n} :~ \exists P_{t^m}, ~\exists Q, Q' \in \PP^{m_2}, \text{ s.t. } \\ \nonumber
& \hspace{1.4cm} d_{L_1}(P_{t^m}, P_{\tau^m}) \le 2\alpha \nonumber \\
& \hspace{1.4cm} \text{and } h(P_{x^n}, P_{t^m} - \alpha Q'+ \alpha Q)  \le  \lambda -  \delta_{n} \}.\nonumber
\end{align}
By letting $n$ go to infinity, we obtain the asymptotic counterpart of the above set, which, for a generic $R \in \PP$, takes the following expression:
\begin{align}
\Gamma_0(R, \lambda, \alpha) ~=~ \big\{ P: ~ &\exists P', Q, Q',  \text{ s.t. } d_{L_1}(P', R) ~\le~ 2\alpha  \nonumber\\  & \text{and }  h_c(P, P' - \alpha Q'+ \alpha Q) ~\le~ \lambda \big\}.
\end{align}
When $L \ne 0$, we obtain:
\begin{align}
\label{set_Gamma_r}
\Gamma(R,\lambda,\alpha,L) ~=~  \{P : ~\exists V \in  \Gamma_0(R,\lambda,\alpha) \text{ s.t. } \text{\em EMD} (P,V) ~\le~ L\}.
\end{align}
With the above definitions, it is straightforward to extend Theorem \ref{theo_as_payoff_si} to the \SIr~ case, thus proving that the set in \eqref{set_Gamma_r} evaluated in $R ~=~ P_X$ represents the indistinguishability region of the \SIr~ game.

\subsection{Security margin and blinding corruption level}

As a last contribution, we are interested in studying the ultimate distinguishability of two sources $X$ and $Y$ in the \SIr~ setting and compare it with the result we have obtained for the \SIa~ case. To do so, we consider the behaviour of the indistinguishability region when $\lambda$ tends to 0. We have:
\begin{align}
\Gamma(P_X,\alpha,L) ~=~  \{P :~ \exists V \in & \Gamma_0 (P_X,\alpha) \text{ s.t. } \text{\em EMD} (P,V) ~\le~ L\},
\label{eq.indist.set_L}
\end{align}
where
\begin{align}
\Gamma_0(P_X, \alpha) ~&=~ \big\{  P  : ~ \exists P', Q, Q' \text{ s.t. } d_{L_1}(P', P_X) ~\le~ 2\alpha \nonumber\\
& \hspace{0.6cm} \text{and }  P ~=~ P' + \alpha(Q - Q') \big\} \nonumber\\
&=~  \big\{P  : ~\exists P' \text{ s.t. } d_{L_1}(P', P_X) ~\le~ 2\alpha  \nonumber\\
& \hspace{0.6cm}\text{and }  d_{L_1}(P,P') ~\le~ 2\alpha \big\}.
\label{eq.indist.set}
\end{align}
The set in \eqref{eq.indist.set} can be equivalently rewritten as
\begin{align}
\Gamma_0(P_X,\alpha) ~=~ \big\{P  :~ d_{L_1}(P, P_X) ~\le~ 4\alpha \big\}.
\label{eq.indist.set2}
\end{align}

To see why, we first notice that set \eqref{eq.indist.set} is contained in \eqref{eq.indist.set2}. Indeed, from the triangular inequality we have that, for any $P'$, $d(P,P_X) ~\le~ d_{L_1}(P,P') + d_{L_1}(P',P_X)$. Then, if $P$ belongs to $\Gamma_0(P_X,\alpha)$ in \eqref{eq.indist.set}, it also belongs to the set in \eqref{eq.indist.set2}.  To see that the two sets are indeed equivalent, it is sufficient to show that the reverse implication also holds. To this purpose, we observe that, whenever $d_{L_1}(P, P_X) ~\le~ 4\alpha$, a type $P^*$ can be found such that its distance from both $P$ and $P_X$ is less or at most equal to $2\alpha$. In fact, by letting $P^* ~=~ \frac{P + P_X}{2}$, we have
\begin{align}
&d_{L_1}(P,P^*) ~=~ d_{L_1}(P^*,P_X) ~=~  \sum_i \bigg|\frac{P(i) - P_X(i)}{2}\bigg| \nonumber\\
&d_{L_1}(P,P_X) ~=~ \sum_i \bigg|P_X(i) - P(i)\bigg| ~=~  2 d_{L_1}(P,P^*).
\end{align}
If $d_{L_1}(P, P_X) ~\le~ 4\alpha$, then, $d_{L_1}(P,P^*) ~=~ d_{L_1}(P^*,P_X)$ $~=~ d_{L_1}(P,P_X)/2 ~\le~ 2\alpha$, permitting us to conclude that the sets in \eqref{eq.indist.set} and \eqref{eq.indist.set2} are equivalent.

Upon inspection of equation \eqref{eq.indist.set2}, we can conclude that, as expected, the indistinguishability region for $L=0$ (and hence, also for the case $L \neq 0$)  is larger than that of the \SIa~ game (see \eqref{Gamma_X_2}), thus confirming that the game with sample replacement is more favourable to the attacker (a graphical comparison between the indistinguishability regions for the two setups is shown in Figure \ref{fig.set_comparison}).
\begin{figure}
\centering \includegraphics[width = 0.55\columnwidth]{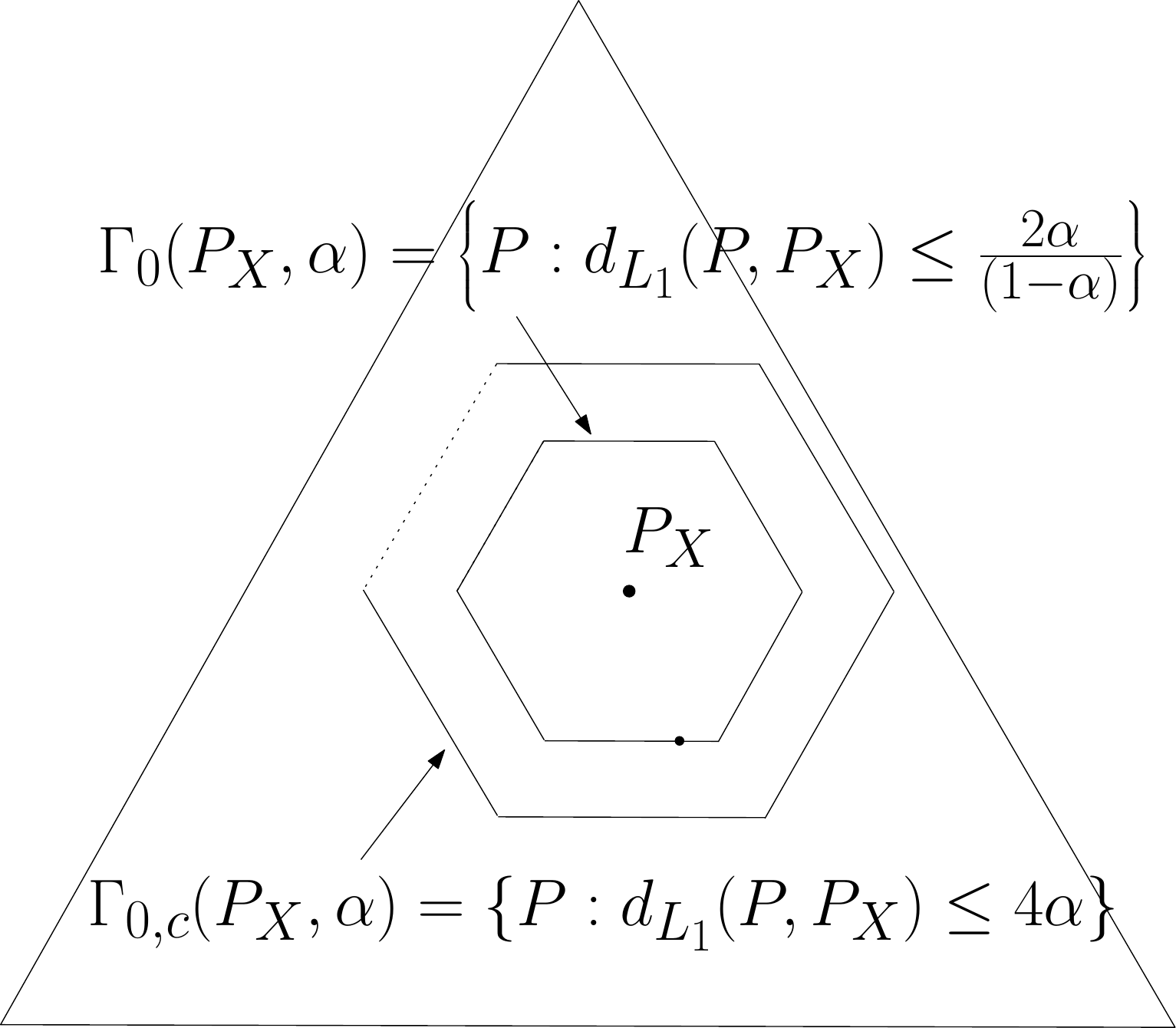}
\caption{Comparison of the indistinguishability regions for the \SIa~  and \SIr~  games with $L = 0$.}
\label{fig.set_comparison}
\end{figure}
As a matter of fact, for the attacker, the advantage of the \SIr~ game with respect to the \SIa~ game depends on $\alpha$.
For small $\alpha$ and for $\alpha$ close to $1/2$, the indistinguishability regions of the two games are very similar, while for intermediate values of $\alpha$ the indistinguishability region of the \SIr~ game is considerably larger than that of the \SIa~ game (the maximum difference between the two regions is obtained for $\alpha \approx 0.3$). When $\alpha = 1/2$ the attacker always wins, since he is able to bring any pmf inside the acceptance region regardless of the game version, while for $\alpha = 0$, we fall back into the source identification game without corruption of the training sequence, thus making the two versions of the game equivalent.

Given two sources $X$ and $Y$, the blinding corruption level value takes the expression:
\begin{equation}
\alpha_b ~=~ \frac{d_{L_1}(P_Y,P_X)}{4}.
\end{equation}
Since $d_{L_1}(P_Y,P_X) \le 2$ for any couple $(P_Y, P_X)$ (the maximum value 2 is taken when the two distribution have disjoint support), the blinding value for the \SIr~ game is lower than the blinding value of  \SIa~ game. The two expressions are identical when the two sources have disjoint support, in which case $\alpha_b = 1/2$.

%Another interesting observation (point to stress) is that the two hostile regions are exactly the same when $\alpha = 1/2$, in which case the attacker always win, being able to bring any pdf inside the acceptance region, either he chooses the to-be-modified samples or not.

When the attacker can also corrupt the test sequence, the {\em ultimate indistinguishability region} of the \SIr game is:
\begin{align}
\Gamma(P_X,\alpha,L) ~=~  \big\{P: \min_{V: \text{\em EMD}(P,V) \le L}  d_{L_1}(V, P_X) ~\le~ 4\alpha \big\}.
\label{eq.preSM}
\end{align}
Starting from \eqref{eq.preSM} we can define the security margin in the \SIr~ setup.

\begin{definition}[Security Margin in the \SIr~ setup]
Let $X \sim P_X$ and $Y \sim P_Y$ be two discrete memoryless sources. The maximum distortion for which the two sources can be reliably distinguished in the \SIr~ setup is called Security Margin and is given by
\begin{align}
\SS \MM_{\alpha}& (P_X, P_Y) = L_{\alpha}^*,
\label{security_margin}
\end{align}
where $L_{\alpha}^*$ is the quantity which satisfies the following relation
\begin{equation}
\min_{V: \text{\em EMD}(P_Y,V) \le L_{\alpha}^*} d_{L_1}(V,P_X) ~=~ 4\alpha,
\end{equation}
if $P_Y \notin \Gamma_0(P_X, \alpha)$, and $L_{\alpha}^* = 0$ otherwise.
%\begin{align}
%\min_{R: \text{\em EMD}(P_Y,R) = L_{\alpha}^*} \sum_{i} \left[R(i) - P_X(i) \right]^+  = \frac{\alpha}{(1 - \alpha)},
%\end{align}
%%
\end{definition}
Considering again the case of two Bernoulli sources and by adopting the same notation of Section \ref{sec_bernoulli}, we have that $\alpha_b = |p - q|/4$, while the security margin is
\begin{equation}
\label{SM_bernoulli}
\SS\MM_{\alpha}(p,q) =  \left\{
\begin{array}{ll}
 |q - p| - 2 \alpha & \alpha ~<~ \alpha_b \\
0 & \alpha ~\ge~ \alpha_b
\end{array}
\right..
\end{equation}
Figure \ref{fig.SMcurve_Bernoulli} plots $\SS\MM_{\alpha}$ as a function of $\alpha$ when $p=0.3$ and $q = 0.7$. The blinding value is $\alpha_b=0.1$ which, as expected, is lower than the value we found for the \SIa~ setup.
\begin{figure}[t!]
\centering \includegraphics[width = 0.55\columnwidth]{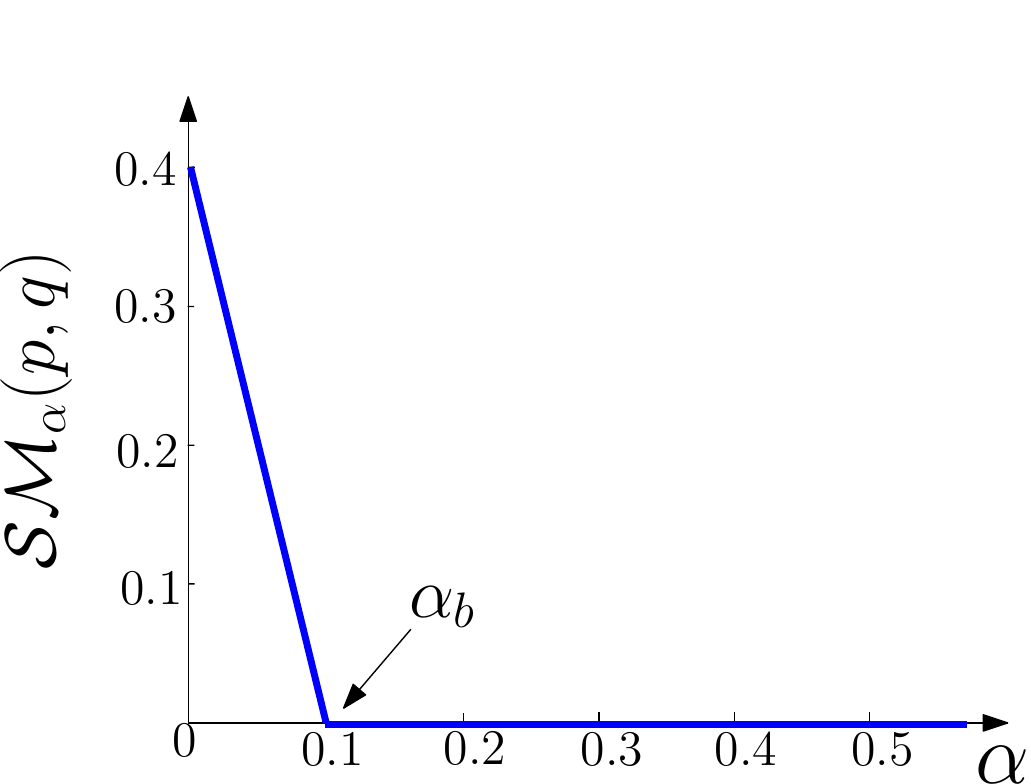}
\caption{Security margin as a function of $\alpha$ for Bernoulli sources with parameters $p = 0.3$ and $q=0.7$ ($\alpha_b = 0.1$).}
\label{fig.SMcurve_Bernoulli2}
\end{figure}

\section{Conclusions}
\label{sec.conc}

We studied the distinguishability of two sources in an adversarial setup when the sources are known through training data, part of which can be corrupted by the attacker himself. We considered two different scenarios. In the first one, the attacker simply adds fake samples to the original training sequence, while in the second one, the attacker replaces a selected subset of training samples with fake ones. We formalised both cases in a game-theoretic setup, then we derived the equilibrium point of the games and analysed the (asymptotic) payoff at the equilibrium.
%All our findings
The result of the game can be summarised in a compact and elegant way by introducing two parameters, namely the Security Margin under corruption of the training sequence, and the blinding corruption level $\alpha_b$, defined as the portion of fake samples the attacker must introduce to make impossible any reliable distinction between the sources. Based on these two parameters, the performance of the two games with corruption of the training data can be easily compared.

Though rather theoretical, our findings can guide more practical researches in several fields belonging to the emerging areas of adversarial signal processing \cite{BarGon13} and secure machine learning \cite{Barreno2010}. In many cases, in fact, the defender must take into account the possibility that the data he is using to tune the system he is working at, or during the learning phase, is corrupted by the attacker.

The analysis carried out in this paper can be extended in several ways, for instance by considering continuous sources, or by assuming that the sources $X$ and $Y$ are not memoryless, but still amenable to be studied by using the method of types \cite{Csi98}. Following the analysis in \cite{BMTWIFS15}, we could also consider a more general setup in which the attacker is active under both $H_0$ and $H_1$.
An interesting generalisation, consists in studying a symmetric setup in which the training and the test sequences can be corrupted by applying the same kinds of processing. For instance, the attacker could be allowed to replace samples in both the training and the set sequences, or he could be allowed to modify the training sequence up to a certain distortion. Other kinds of attacks to the training data could also be considered, like sample removal with no addition of fake samples. As a matter of fact, the kind of attack strongly depends on the application scenario, and it is arguable that the availability of a large variety of theoretical models would help bridging the gap between theory and practice.

\section*{Acknowledgment}

This work has been partially supported by a research sponsored by DARPA and Air Force Research Laboratory (AFRL) under agreement number FA8750-16-2-0173. The U.S. Government is authorised to reproduce and distribute reprints for Governmental purposes notwithstanding any copyright notation thereon. The views and conclusions contained herein are those of the authors and should not be interpreted as necessarily representing the official policies or endorsements, either expressed or implied, of DARPA and Air Force Research Laboratory (AFRL) or the U.S. Government.

\bibliographystyle{IEEEtran}
\bibliography{HT_CorruptedTr new formulation}

%%%%%%%%%%%%%%%%%%%%%%%%%%%%%%%%%%%%%%%%%%%%%%%%%%%%%%%%%%%%%%%%%%%%%%%%%%%%%%%%%%%%%%%%%%%%%%%%%%%%%%%%%%%%%%%%%%%%%%%%%%%%%%%%%%%%%%%%%

\numberwithin{equation}{section}
%\appendix[Main proofs]
\appendix

\renewcommand{\theequation}{A\arabic{equation}}

\subsection{Generalized Sanov's theorem}
\label{sec.appendix.Sanov}

Let  us consider a sequence of $n$ i.i.d. discrete random variables taking values in a finite alphabet $\XX$ and distributed according to a pmf $P$. We denote with $P_n$ the empirical pmf of the sequence.
Let $E \subseteq \mathcal{P}$ be a set of pmf's.
%We define $E_n = E \cap \mathcal{P}$.
Sanov's theorem \cite{CandT,Sanov,Dembo2009} states that
\begin{align}
 \inf_{Q \in \text{$E$}} \DD(Q||P)
~\le~ & - \underset{n \rightarrow \infty}{\lim ~\sup} ~\frac{1}{n} \log P(P_n \in E)\nonumber\\
~\le~ & - \underset{n \rightarrow \infty}{\lim ~\inf} ~\frac{1}{n} \log P(P_n \in E)\nonumber\\
~\le~ &   \inf_{Q \in \text{{\em int} $E$}} \DD(Q||P),
\label{eq.Sanov1}
\end{align}
%
% ATT: nell'ultima diseguaglianza ci vuole la chiusura {\em cl} se si lavora con spazi funzionali (setting continuo)....
where {\em int S} denote the interior part of the set $S$.

%If, in addition, {\em E = cl(int(E))} or $E$ corresponds to the open set/the internal/inner part {\em E = int(cl(int(E)))} (i.e., to the interior part of {\em E = cl(int(E))}),
When {\em cl}($E$) = {\em cl}({\em int}($E$))\footnote{{\em cl}($E$) denotes the closure of $E$. Clearly, {\em cl}($E$) $\equiv$ $E$ if $E$ is a closed set.}, or,  $E ~\subseteq~ ${\em cl}({\em int}($E$)),
the left and right-hand side of \eqref{eq.Sanov1} coincide and we get the exact rate:
\begin{equation}
-\lim_{n \rightarrow \infty} ~\frac{1}{n} \log P(P_n \in E) ~=~ \inf_{Q \in E} \DD(Q||P).
\label{Sanov}
\end{equation}
If we define the set $E_n = E \cap \mathcal{P}^n$, we have: $P(P_n \in E) = P(P_n \in E_n)$ and we can rewrite Sanov's theorem as:
\begin{align}
 \inf_{Q \in \text{$E$}} \DD(Q||P)
~\le~ & - \underset{n \rightarrow \infty}{\lim ~\sup} ~\frac{1}{n} \log P(P_n \in E_n)\nonumber\\
~\le~ & - \underset{n \rightarrow \infty}{\lim ~\inf} ~\frac{1}{n} \log P(P_n \in E_n)\nonumber\\
~\le~ &   \inf_{Q \in \text{{\em int} $E$}} \DD(Q||P),
\label{eq.Sanov2}
\end{align}
Note that, by construction, we have {\em cl}($E$) = {\em cl}($\cup_n E_n$).

In the following, we extend the formulation of Sanov's theorem given in \eqref{eq.Sanov2} to more general sequences of sets $E_n$ for which it does not necessary hold that $E_n = E \cap \mathcal{P}^n$ for some set $E$.

We start by %providing some necessary definitions.
introducing the notion of convergence for sequences of subsets due to Kuratowsky, which is a more general notion of convergence with respect to the one based on Hausdorff distance. Let $(S, d)$ be a metric space.
We first provide the definition of {\em lower closed limit} or Kuratowski limit inferior \cite{kuratowski1968topology}.
\begin{definition}
A point $p$ belongs to the lower limit $\underset{n \rightarrow \infty}{Li} K_n$ (or simply ${Li} K_n$) of a sequence of sets $K_n$, if every neighborhood of $p$ intersects all the $K_n$'s from a sufficiently great index $n$ onward.

Given the above definition, the expression $p ~\in \underset{n \rightarrow \infty}{Li} K_n$ is equivalent to the existence of a sequence of points $\{p_n\}$ such that:
\begin{equation}
\text{$p ~= \lim_{n \rightarrow \infty} p_n$, \quad $p_n \in K_n$.}
\label{def_Kur_1}
\end{equation}
Stated in another way, ${Li} K_n$ is the set of the accumulation points of sequences in $K_n$.
As an alternative, equivalent, definition we can let:
\begin{equation}
\underset{n \rightarrow \infty}{Li} K_n ~=~ \{p ~\in~ X \text{ s.t. } \underset{n \rightarrow \infty}{\lim ~\sup } \hspace{0.1cm} d(x, K_n) ~=~  0\}.
\label{def_Kur_2}
\end{equation}
\end{definition}
%
%
%Kuratowsky proved that the formula $p \in {Li}_{n \rightarrow \infty}E_n$ is equivalent to the existence of a sequence of points $\{p_n\}$ such that:
%%
%%\begin{equation}
%$p = \lim_{n \rightarrow \infty} p_n$ and $p_n \in E_n$ (see \cite{kuratowski1968topology}). Put differently, ${Li}_{n \rightarrow \infty}E_n$, simply denoted ${Li} E_n$, is the set of the accumulation points of sequences in $E_n$.
%%\end{equation}

Similarly, we have the following definition of
{\em upper closed limit} or Kuratowski limit superior \cite{kuratowski1968topology}.
\begin{definition}
A point $p$ belongs to the upper limit $\underset{n \rightarrow \infty}{Ls} K_n$ (or simply ${Ls} K_n$) of a sequence of sets $K_n$, if every neighborhood of $p$ intersects an infinite number of terms in $K_n$.

The expression $p ~\in~ {Ls}_{n \rightarrow \infty}K_n$ is equivalent to the existence of a subsequence of points $\{p_{k_n}\}$ such that
$$\text{$k_1 < k_2 < \dots$, \quad $p ~=~ \lim_{n \rightarrow \infty} p_{k_n}$, \quad $p_{k_n} \in ~K_{k_n}$}.$$
As an alternative, equivalent, definition we can let:
\begin{equation}
\underset{n \rightarrow \infty}{Ls} K_n ~=~ \{p ~\in~ X \text{ s.t. } \underset{n \rightarrow \infty}{\lim ~\inf } \hspace{0.1cm} d(x, K_n) ~=~  0\}.
\end{equation}
\end{definition}

It can be proven that the Kuratowski limit inferior and superior are always closed set (see \cite{kuratowski1968topology}).

Given the above, we can state the following:
\begin{definition}
The sequence of sets $\{K_n\}$ is said to be convergent to $K$ in the sense of Kuratowski, that is $K_n \overset{K}{\rightarrow}~ K$, if $Li K_n = K = Ls K_n$, in which case we write $K = Lim K_n$.
\end{definition}

We observe that Kuratowski convergence is weaker than convergence in Hausdorff metric; in fact, given a sequence of closed sets $\{K_n\}$, $K_n \overset{H}{\rightarrow}~ K$ implies $K_n \overset{K}{\rightarrow}~ K$ \cite{Salinetti1979}. For compact metric spaces, the reverse implication also holds and the two kinds of convergence coincide.

In this work, we are interested in the space $\PP$ of probability mass functions defined over a finite alphabet $\XX$, i.e., the probability simplex in $\mathbb{R}^{|\XX|}$, equipped with the $L_1$ metric.
Being $\PP$ a closed subset of $\mathbb{R}^{|\XX|}$, $\PP$ is a complete set. In addition,
with the $L_1$ metric, $\PP \in \mathcal{L}(\mathbb{R}^{|\XX|})$, that is, $\PP$ is bounded.
The space $(\PP, d_{L_1})$, then, is a compact metric space
%\footnote{\BTcomm{A rigore, per la compattezza dello spazio in tutti gli spazi (euclidei e non), serve che $\PP$ sia totalmente limitato e non solo limitato.....Non so quanto sia il caso di scendere nei dettagli su questo punto. Che dici?} \MB{I would leave it as is}}
and then, for our purposes, Kuratowski and Hausdorff convergence are equivalent.

We are now ready to prove the following generalisation of Sanov's theorem:

\begin{theorem}[Generalized Sanov's theorem]
\label{theo.extended_Sanov}
    Let $\{E_{(n)}\}$
    be a sequence of sets in $\PP$, such that $Li (E_{(n)} \cap \PP^n) \ne ~\emptyset$.
Then:
%%
%\begin{align}
%- \inf_{Q \in \text{{\em int} $E$}} \DD(Q||P) \le &\underset{n \rightarrow \infty}{\lim \inf} \frac{1}{n} \log P(\hat{P}_n \in E_{(n)})\nonumber\\
%\le & \underset{n \rightarrow \infty}{\lim \sup} \frac{1}{n} \log P(\hat{P}_n \in E_{(n)})\nonumber\\
%\le & - \inf_{Q \in \text{$E$}} \DD(Q||P),
%\label{Sanov1}
%\end{align}
%%
\begin{align}
 \min_{Q \in \text{ $Ls E_{(n)}$}} \DD(Q||P)
~\le~ & - \underset{n \rightarrow \infty}{\lim \sup} ~\frac{1}{n} \log P(P_n \in E_{(n)})\nonumber\\
\le~ & - \underset{n \rightarrow \infty}{\lim \inf} ~\frac{1}{n} \log P(P_n \in E_{(n)})\nonumber\\
\le~ &   \min_{Q \in \text{ $Li$ ($E_{(n)} \cap \PP^n$)}} \DD(Q||P),
\label{Sanov1}
\end{align}
If, in addition,
$Ls E_{(n)} = Li (E_{(n)} \cap \PP^n)$,
%\footnote{Note that such condition implies that $Ls E_{(n)} = Li E_{(n)}$, that is that $E_{(n)}$ converges in the sense of Kuratowsky. However, we stress that this is only a necessary condition, that is, the Kuratowsky convergence of $E_{(n)}$ to $Lim E_{(n)}$ alone is not sufficient for the existence of the limit in \eqref{extended_Sanov}. } %$E_{(n)}\overset{K}{\rightarrow}~ E$
the generalized Sanov's limit exists as follows:
\begin{equation}
-\lim_{n \rightarrow \infty} ~\frac{1}{n} \log P(P_n \in E_{(n)}) ~= \min_{Q \in Lim E_{(n)}} \DD(Q||P).
\label{extended_Sanov}
\end{equation}
\end{theorem}

\begin{proof}
We first prove the expression for the lower bound.
%in (\ref{extended_Sanov}) is a lower bound.
Let $E_n ~=~ E_{(n)} \cap \PP^n$. We have:
\begin{eqnarray}
     P(E_{(n)}) & = & \sum_{Q \in E_n}P_X(T(Q)) \nonumber \\
     %& {\le} & \sum_{P \in \Pi_n} 2^{-n \DD(P || P_X)} \nonumber \\
     & {\le} & (n + 1)^{|\XX|} 2^{-n \min_{Q \in E_n} \DD(Q ||
     P)}\nonumber\\
     & {\le} & (n + 1)^{|\XX|} 2^{-n \inf_{Q \in E_{(n)}} \DD(Q ||
    P)}\nonumber\\
         & = & (n + 1)^{|\XX|} 2^{-n \min_{Q \in \text{\em cl}(E_{(n)})} \DD(Q ||
    P)}.
%      & {\le} & (n + 1)^{|\XX|} 2^{-n \inf_{Q \in \bigcup_{i \ge n} E^{(i)}} \DD(Q ||
%     P)}.
\label{lower_bound}
\end{eqnarray}
%
%where $\overline{s}$ is an alternative form for $\text{\em cl}(s)$.
In the last inequality we exploited the fact that, being each $E_{(n)}$ a bounded set of $\PP$, and $\DD$ lower bounded in $\PP$, the infimum over $E_{(n)}$ corresponds to the minimum over its closure. By taking the logarithm of each side and dividing by $n$, we get:
\begin{equation}
\frac{1}{n} \log P(E_{(n)})   ~\le~  - \min_{Q \in \text{\em cl}(E_{(n)})} \DD(Q ||
    P) + \frac{\log(n+1)^{|\XX|}}{n},
    \label{lower_bound_2}
\end{equation}

We now prove that, for any $\delta$ and for sufficiently large $n$, we have
\begin{equation}
\min_{Q \in \text{\em cl}(E_{(n)})} \DD(Q || P) ~\ge~ \min_{Q \in Ls E_{(n)}} \DD(Q || P) ~-~ \delta .
     \label{relation_Ls}
\end{equation}
First, according to the properties of the limit superior, $Ls E_{(n)} = Ls (\text{\em cl}(E_{(n)}))$ \cite{kuratowski1968topology}, hence proving \eqref{relation_Ls} is equivalent to showing that:
\begin{equation}
\min_{Q \in \text{\em cl}(E_{(n)})} \DD(Q || P) ~\ge~ \min_{Q \in Ls (\text{\em cl}(E_{(n)}))} \DD(Q || P) ~-~ \delta .
     \label{relation_Ls_bis}
\end{equation}
Let $Q_n$ be the sequence of points achieving the minimum of the left-hand side of \eqref{relation_Ls_bis} (for simplicity we assume that the minimum is unique, the extension to a more general case being straightforward). Let $Q_{n(j)}$ be a subsequence of $Q_n$ formed only by the elements of $Q_n$ that do not belong to $Ls (\text{\em cl}(E_{(n)}))$\footnote{ $n(i) ~>~ n(j), \forall i ~>~ j$}. If the number of elements in $Q_{n(j)}$ is finite, then for $n$ large enough $Q_n ~\in~ Ls (\text{\em cl}(E_{(n)}))$ and eq. \eqref{relation_Ls_bis} is verified with $\delta = 0$. If the number of elements in $Q_{n(j)}$ is infinite, then, due to the boundedness of $\PP$, the elements of $Q_{n(j)}$ must have at least one accumulation point (Bolzano-Weierstrass theorem). Let $A_i$'s be the accumulation points of $Q_{n(j)}$. By definition of $Ls$, all $A_i$'s belong to $Ls (\text{\em cl}(E_{(n)}))$. In addition, for any radius $\rho$, from a certain $j$ on, all the points in $Q_{n(j)}$ belong to $\mathcal{R} = \bigcup_i \BB(A_i, \rho)$\footnote{$\BB(A_i, \rho)$ is a ball with radius $\rho$ centred in $A_i$.}. For large enough $n$, then we have:
\begin{align}
\min_{Q \in \text{\em cl}(E_{(n)})} \DD(Q || P) ~&\ge~ \min_{Q \in Ls (\text{\em cl}(E_{(n)})) \cup \mathcal{R}} \DD(Q || P) \\
& \ge \min_{Q \in Ls (\text{\em cl}(E_{(n)}))} \DD(Q || P) ~-~ \delta, \nonumber
\label{eq.gsanovMB}
\end{align}
where the second inequality derives from the continuity of the $\DD$ function and the arbitrariness of $\rho$.

By inserting equation \eqref{relation_Ls} in \eqref{lower_bound_2}, we have that, for large $n$,

\begin{equation}
\frac{1}{n} \log P(E_{(n)})  \le  - \hspace{-0.2cm}\min_{Q \in Ls E_{(n)}} \DD(Q ||P) + \frac{\log(n+1)^{|\XX|}}{n} + \delta,
\end{equation}
and hence, by the arbitrariness of $\delta$,
\begin{align}
  - \underset{n \rightarrow \infty}{\lim \sup} ~\frac{1}{n} \log P(E_{(n)}) ~\ge~ \min_{Q \in Ls E_{(n)}} \DD(Q ||
     P).
\end{align}

We now pass to the upper bound. Let $Q^*$ be a point achieving the minimum of the divergence over the set $Li E_n$. By definition of limit inferior, there exists a sequence of points $\{Q_n\}$, $Q_n \in E_n$ such that $Q_n \rightarrow Q^*$ as $n \rightarrow \infty$.
Then, by exploiting the continuity of $\DD$, it follows that:
\begin{align}
\DD(Q_n|| P) ~\le~ D(Q^*|| P) ~+~ \gamma,
\end{align}
where $\gamma$ can be made arbitrarily small for large $n$. We can then write:
\begin{eqnarray}
    P(E_{(n)}) & = & \sum_{Q \in E_n} P(T(Q)) \nonumber \\
    & \ge & P(T(Q_n)) ~\ge~  \frac{2^{-n \DD(Q_n || P)}}{(n + 1)^{|\XX|}}.
\label{upper_bound}
\end{eqnarray}
Hence, we get
\begin{align}
\frac{1}{n} \log P(E_{(n)}) & ~\ge~ - \DD(Q_n ||  P) ~-~ |\XX|\frac{\log(n + 1)}{n},\nonumber\\
& ~\ge~ - \DD(Q^* ||  P) ~-~ \gamma ~-~ |\XX|\frac{\log(n + 1)}{n},\nonumber\\
& ~\ge~ - \min_{Q \in Li E_n} \DD(Q || P) ~-~ \gamma ~-~ |\XX|\frac{\log(n + 1)}{n},
%\frac{1}{n} \log P(E_n) & \ge - \DD(Q^* ||
%     P) - \nu - |\XX|\frac{\log(n + 1)}{n},\nonumber\\
%     &  = - \inf_{Q \in \text{\em int}(E)} \DD(Q ||
%     P) - \nu - |\XX|\frac{\log(n + 1)}{n},
\end{align}
and then, by the arbitrariness of $\gamma$,
\begin{align}
- ~\underset{n \rightarrow \infty}{\lim \inf} ~\frac{1}{n} \log P(E_{(n)}) ~\le~ \min_{Q \in Li E_n} \DD(Q ||
     P),
     \label{liminf}
\end{align}
which concludes the proof of the first part (relation \eqref{Sanov1}).

For the proof of the second part, we observe that, when $Ls E_{(n)} = Li (E_{(n)} \cap \PP^n)$, the two bounds in \eqref{Sanov1} coincides.
Moreover, the following chain of inclusions holds, $Li E_{(n)} ~\subseteq~ Ls E_{(n)} ~=~ Li (E_{(n)} \cap \PP^n) ~\subseteq~ Li E_{(n)}$, and then $Li E_{(n)} ~=~ Ls E_{(n)} ~=~ Lim E_{(n)}$, yielding \eqref{extended_Sanov}.
\end{proof}
We observe that, in general, the Kuratowski convergence of $E_{(n)}$ is a {\em necessary} condition for the existence of the generalized Sanov limit in \eqref{extended_Sanov}, but it is not sufficient. In fact, we could have $Li E_{(n)} ~\supseteq~ Li(E_{(n)} \cap \PP^n)$, in which case the lower and upper bound in \eqref{Sanov1} do not coincide.
It is also interesting to notice that when $E_{(n)} ~\in~ \PP^n$ is a sequence of sets in $\PP^n$, then Sanov's limit holds whenever $E_{(n)} \overset{K}{\rightarrow}E$ for some set $E$, or, by exploiting the compactness of $\PP$, $E_{(n)} \overset{H}{\rightarrow} E$. Based on the above observation, we can state the following corollary:

\begin{corollary}
\label{cor.extended_Sanov}

Let $E_{(n)}$ be a sequence of sets in $\PP^n$, such that $E_{(n)} \overset{H}{\rightarrow} E$.
Then:
\begin{equation}
-\lim_{n \rightarrow \infty} ~\frac{1}{n} ~\log ~P(P_n \in E_{(n)}) ~=~ \min_{Q \in E} ~\DD(Q||P).
\label{corollary_extended_Sanov}
\end{equation}
\end{corollary}

%
% Ho tolto la parte sottostante perché l'appendice è già molto lunga e l'osservazione che il teorema generalizzato si riduce al teorema classico non è affatto sorprendente, e inoltre non ci serve nel corso dell'articolo.
%
%{\em Observation}.
%When the sequence $\{E_{(n)}\} = E$ $\forall n$ (or from a certain $n$ on), the generalized Sanov's theorem corresponds to the {\em classical} Sanov's theorem. In fact, we have that $Ls E_{(n)} = E$, while $Li (E_n) ~=~ Li (E \cap \PP_n)$, i.e., the set of all the accumulation points of sequences in $E \cap \PP$. Since $Li (E \cap \PP_n) \supseteq int E$\footnote{It is easy to show that every $p \in int(E)$ is accumulation point for a sequence in $E \cap \PP_n$.}, we can write the Sanov's bounds:
%%
%\begin{align}
% \inf_{Q \in E} \DD(Q||P)
%\le & - \underset{n \rightarrow \infty}{\lim \sup} \frac{1}{n} \log P(P_n \in E)\nonumber\\
%\le & - \underset{n \rightarrow \infty}{\lim \inf} \frac{1}{n} \log P(P_n \in E)\nonumber\\
%\le &   \inf_{Q \in \text{ $Li$ ($E \cap \PP_n$)}} \DD(Q||P),\nonumber\\
%\le &   \inf_{Q \in \text{\em int} E} \DD(Q||P).
%\label{Sanov1new}
%\end{align}

\subsection{Regularity properties of the set of admissible maps}
\label{sec.appendix.cont_MAP}

To prove the theorems on the asymptotic behaviour of the payoff in the two versions of the source identification game studied in this paper, we need to prove some regularity theorems on the set of admissible maps.

To start with, we need to define a distance between transportation maps, that is a function $d_s : ~\mathbb{R}^{|\XX| \times |\XX|} \times \mathbb{R}^{|\XX| \times |\XX|} \rightarrow \mathbb{R}^{+}$. In accordance with the rest of the paper, let us choose the $L_1$ distance, that is, given two maps ($S_{PV}, S_{QR}$), we define $d_s (S_{PV}, S_{QR}) = \sum_{i,j} |S_{PV}(i,j) - S_{QR}(i,j)|$.

Our first result regards the regularity of $\mathcal{A}(L, P)$ as a function of $P$.

\begin{lemma}
Let $P \in \PP$ and let $P'$ be any pmf in the neighbourhood of $P$ of radius $\tau$, i.e., $P' \in \mathcal{B}(P,\tau)$. Then
$$
\delta_H(\AA(L,P), ~\AA(L,P')) ~\le~ \tau
$$
and hence $\underset{\tau \rightarrow 0}{\lim} \delta_H(\AA(L,P), \AA(L,P')) ~=~ 0$, uniformly in $\PP$.

\noindent Moreover, if we insist that $P' \in \PP^n$, the following result holds: $\forall \varepsilon > 0, ~\exists \tau^*$ and $n^*$ such that $\forall \tau < \tau^*$ and $n > n^*$,
$$
\delta_H(\AA(L,P), ~\AA^n(L,P')) ~\le~ \varepsilon \quad \forall P' \in \BB(P,\tau) \cap \PP^n, ~\forall P \in \PP.
$$

\label{regularity_adm_set}
\end{lemma}

\begin{proof}

From a general perspective, the lemma follows from the fact that
$\mathcal{A}^n(L,P_{y^n})$ (and $\mathcal{A}(L,P)$) is built by imposing a number of linear constraints on the admissible transportation maps (see eq. \eqref{eq.admissiblemap1}), i.e. $\mathcal{A}(L,P)$ is a convex polytope \cite{Bert97,Boyd04}. By considering a $P'$ close to $P$, we are perturbing the vector of the known terms of the linear constraints which defines the admissibility set. Instead of invoking the above general principle, in the following we give an explicit proof of the lemma.

Given $P \in \PP$ and $P' \in \mathcal{B}(P, \tau)$, let $\tau(i) = P(i) - P'(i)$ be the excess (or defect) of mass of $P$ with respect to $P'$ in bin $i$. For any map in $\mathcal{A}(L,P)$, we can choose a map $S_{P'V'}$ that works as follows: for the bins $i$ such that $\tau(i) \le 0$, let $S_{P'V'}(i,j) = S_{PV}(i,j)$ for $j \neq i$, while for $j=i$, we let $S_{P'V'}(i,j) = S_{PV}(i,j) + | \tau(i) |$. For the bins $i$ for which $\tau(i) > 0$, we first sort the index set $\{j: S_{PV}(i,j) \neq 0\}$ in decreasing order with respect to the amount of distortion introduced per unit of mass delivered from $i$ to $j$ ($d(i,j)$). Then, starting from the first index in the ordered list, we let $S_{P'V'}(i,j) = \max(0, ~S_{PV}(i,j) - \tau(i))$. If $S_{P'V'}(i,j) = 0$, we update $\tau(i)$ to a new value $\tau'(i) = \tau(i) - S_{PV}(i,j)$, and iterate the previous procedure by subtracting the updated value of $\tau'(i)$ from the second $S_{PV}(i,j)$ in the list. This procedure goes on until the subtraction gives $S_{P'V'}(i,j) \neq 0$, that is when we have removed all the excess mass from the $i$-th row of $S_{PV}(i,j)$.

It is easy to see that the map built in this way satisfies the distortion constraint, in fact, by construction the distortion associated to $S_{P'V'}$ is less than that introduced by $S_{PV}$. Then, $S_{P'V'} \in \mathcal{A}(L, P')$.
In addition, by construction, $\sum_{j} |S_{P'V'}(i,j) - S_{PV}(i,j)| \le | \tau(i) |$, and hence $\sum_{ij} |S_{P'V'}(i,j) - S_{PV}(i,j)| \le \tau$. Accordingly, we have:
\begin{align}
\delta_{\AA(L,P)}(&\AA(L,P')) = \\
&\max_{S_{PV} \in \AA(L,P)} \min_{S_{P'V'} \in \AA(L,P')} d_s(S_{PV}, S_{P'V'}) ~\le~ \tau \nonumber
\end{align}
since, as we have shown with the preceding construction, the inner minimum is always lower or equal than $\tau$. By repeating the same argument exchanging the role of $\AA(L,P)$ and $\AA(L,P')$, we find that $\delta_H(\mathcal{A}(L, P'), \mathcal{A}(L, P)) \le \tau$, thus concluding the first part of the proof.

In the second part of the lemma, we require that $P' \in \PP^n$ and that the map produces a sequence in $\PP^n$. The proof is easily achieved by exploiting the first part of the lemma according to which for any map $S_{PV}$ in $\AA(L,P)$, we can find a map $S_{P'V'}$ in $\AA(L,P')$ which is arbitrarily close to $S_{PV}$, and then approximating $S_{P'V'}$ with a map $S^n_{P'V'} \in \AA^n(L,P')$. Due to the density of rational numbers in real numbers, such an approximation can be made arbitrarily accurate by increasing $n$, thus completing the proof.
\end{proof}

Given a transformation $S_{PV}$ mapping $P$ into $V$, Lemma \ref{regularity_adm_set} states that, for any pmf $P'$ close to $P$, we can find a map $S_{P'V'}$ close to $S_{PV}$. The following theorem extends such a result to the pmf resulting from the application of the mapping.

\begin{theorem}

Let $P \in \PP$, and let $P'$ be any pmf in the neighbourhood of $P$ of radius $\tau$, i.e., $P' \in \mathcal{B}(P,\tau)$. Let  $S_{PV} \in \mathcal{A}(L,P)$.
Then, we can always find a map $S_{P'V'} \in \mathcal{A}(L, P')$ such that $V' \in \mathcal{B}(V, \tau)$.
%, and then can be made arbitrarily small for small enough $\tau$.

Similarly, for any $\varepsilon > 0$, there exist $\tau^*$ and $n^*$ such that $\forall$ $\tau < \tau^*$ and $n > n^*$, given a $P \in \PP$, a map $S_{PV} \in \mathcal{A}(L, P)$ and $P' \in \PP^n \cap \BB(P,\tau)$,  we can find a map $S_{P'V'}^n$ in  $\mathcal{A}^n(L, P')$ such that $V'_n \in \BB(V, \varepsilon) \cap \PP^n$.
\label{theo_behavior_S}
\end{theorem}
\begin{proof}
For any two maps $S_{PV}$ and $S_{P'V'}$, we have:
\begin{align}
V'(j) &~=~ \sum_{i} S_{P'V'}(i,j) \nonumber \\
&~=~ \sum_{i}  (S_{PV}(i,j) +  (S_{P'V'}(i,j) - S_{PV}(i,j))) \nonumber\\
& ~\le~ V(j) + \sum_i | S_{P'V'}(i,j) - S_{PV}(i,j) |,
\end{align}
and
\begin{align}
V'(j) &~=~ \sum_{i} S_{P'V'}(i,j) \nonumber \\
&~=~ \sum_{i}  (S_{PV}(i,j) +  (S_{P'V'}(i,j) - S_{PV}(i,j))) \nonumber\\
& ~\ge~ V(j) - \sum_i | S_{P'V'}(i,j) - S_{PV}(i,j) |,
\end{align}
yielding:
\begin{equation}
|V'(j) -V(j)| ~\le~ \sum_i | S_{P'V'}(i,j) - S_{PV}(i,j) |.
\end{equation}
By summing over $j$ and exploiting Lemma \ref{regularity_adm_set}, we can choose $S_{P'V'}$ so that:

\begin{align}
\sum_j |V'(j) - V(j)| & ~\le~ \sum_{i,j} | S_{P'V'}(i,j) - S_{PV}(i,j) | \nonumber \\
& ~\le~ \delta_H(\mathcal{A}(L, P'), ~\mathcal{A}(L, P)) ~\le~ \tau,
\end{align}
and hence $V' \in \BB(V, |\tau)$.

Similarly to the second part of  Lemma \ref{regularity_adm_set}, the second part of the theorem follows immediately from the density of rational numbers in the real line.

\end{proof}

\end{document}